\documentclass[12pt]{article}

\usepackage{pifont}
\usepackage{natbib}
\usepackage{geometry}
\usepackage{fleqn}
\usepackage{hyperref}
\usepackage{graphicx}
\usepackage{amsmath}
\usepackage{amsfonts}
\usepackage{amssymb}
\usepackage{natbib}
\usepackage[T1]{fontenc}

\usepackage{latexsym}
\usepackage{times}
\usepackage{setspace}
\usepackage[bf,footnotesize]{caption}
\usepackage[displaymath]{lineno}
\usepackage{color}
\usepackage{eurosym}

\usepackage{yhmath}

\usepackage{epstopdf} 
\usepackage{epsfig}

\usepackage[T1]{fontenc}
\usepackage{authblk}

\def\ds{\displaystyle}

\def\eps{\epsilon} 

\newcommand{\Frac}[2]{\ds\frac{\ds #1}{\ds #2}}

\newtheorem{theorem}{Theorem}

\newtheorem{lemma}[theorem]{Lemma}

\newenvironment{proof}[1][Proof]{\noindent\textbf{#1.} }{\ \rule{0.5em}{0.5em}}

\title{Does the ``uptick rule'' stabilize the stock market? Insights from Adaptive Rational Equilibrium Dynamics}

\author[*]{F. Dercole\thanks{fabio.dercole@polimi.it}}
\author[**]{ D. Radi\thanks{davide.radi@unibg.it}}
\affil[*]{\emph{\small Department of Electronics and Information, Politecnico of Milan, Italy}}
\affil[**]{\emph{\small Department of Management, Economics and Quantitative Methods, University of Bergamo}}

\date{}

\begin{document}
\maketitle
\begin{abstract}
This paper investigates the effects of the ``uptick rule'' (a short selling regulation formally known as rule 10a-
1) by means of a simple stock market model, based on the ARED (adaptive rational equilibrium dynamics) modeling framework, where heterogeneous and adaptive beliefs on the future prices of a risky asset were first shown to be responsible for endogenous price fluctuations.

The dynamics of stock prices generated by the model, with and without the uptick-rule restriction, are analyzed by pairing the classical fundamental prediction with beliefs based on both linear and nonlinear technical analyses. The comparison shows a reduction of downward price movements of undervalued shares when the short selling restriction is imposed. This gives evidence that the uptick rule meets its intended objective. However, the effects of the short selling regulation fade when the intensity of choice to switch trading strategies is high. The analysis suggests possible side effects of the regulation on price dynamics.
\end{abstract}



\section{Introduction}

\bigskip
Short selling is the practice of selling financial instruments that have been
borrowed, typically from a broker-dealer or an institutional investor, with
the intent to buy the same class of financial instruments in a future period
and return them back at the maturity of the loan.

By short selling, investors open a so-called ``short position'', that is technically equivalent to holding a
negative amount of shares of the traded asset, with the expectation that the
asset will recede in value in the next period.
At the closing time
specified in the short selling contract, the debt is compounded
with interest which occurred during the period of the financial operation, for
this reason short sellers prefer to close the short position and reopen
a new one with the same features, rather than extending the position over the closing time (see, e.g., \citep{Hull}).

A short position is the counterpart of the (more conventional) ``long position'', i.e. buying a security
such as a stock, commodity, or currency, with the expectation that the asset
will rise in value.

Short selling is considered the father of the modern derivatives and, as
such, it has a double function: it can be used as an insurance device, by
hedging the risk of long positions in related stocks thus allowing risky
financial operations, or for speculative purposes, to profit from an
expected downward price movement.
Moreover, financial speculators can sell short stocks in an effort to drive
down the related price by creating an imbalance of sell-side interest, the
so called ``bear raid'' action. This
feedback may lead to the market collapse, and has indeed been observed
during the financial crises of 1937 and 2007, see, e.g.,
\citep{MisraLagiBarYam2012}.

Many national authorities have developed different kinds of short selling
restrictions to avoid the negative effect of this financial practice%
\footnote{%
It is worth mentioning that short sale restrictions are nearly as old as
organized exchanges. The first short selling regulation was enacted in 1610
in the Amsterdam stock exchange. For a review of the history of short sale
restrictions, see Short History of the Bear, Edward Chancellor, October 29,
2001, copyright David W. Tice and Co.}.
Most of the regulations are based on ``price tests'', i.e., short selling is
allowed or restricted depending on some tests based on recent price movements. The
best known and most widely applied of such regulations is the so-called
``{\it uptick rule}'', or rule 10a-1, imposed in 1938
by the U.S. Securities and Exchange Commission%
\footnote{%
The rule was originally introduced under the Securities Exchange Act of 1934.}
(hereafter SEC) to protect investors and was in force until 2007. This rule
regulated short selling into all U.S. stock markets and in the Toronto Stock
Exchange.
Other financial markets, like the London Stock Exchange and the Tokyo Stock Exchange,
have different or no short selling restrictions (for a summary of short sale
regulations in approximately 50 different countries see \citep{BrisGoetzmannZhou2007}).

The uptick rule originally stated that short sales are allowed only on an
uptick, i.e., at a price higher than the last reported transaction price.
The rule was later relaxed to allow short sales to take place on a
zero-plus-tick as well, i.e., at a price that is equal to the last sale price but
only if the most recent price movement has been positive. Conversely, short
sales are not permitted on minus- or zero-minus-ticks, subject to narrow
exceptions%
\footnote{%
In the Canadian stock markets, the tick test was introduced under rule 3.1 of UMIR
(Universal Market Integrity Rules). It prevents short sales at a price that is less
than the last sale price of the security.}.

In adopting the uptick rule, the SEC sought to achieve three objectives%
\footnote{%
Quoted from the Securities Exchange Act Release No. 13091 (December 21,
1976), 41 FR 56530 (1976 Release).}:

\begin{itemize}
\item[(i)] \textit{allowing relatively unrestricted short selling in an
advancing market;}

\item[(ii)] \textit{preventing short selling at successively lower prices,
thus eliminating short selling as a tool for driving the market down; and}

\item[(iii)] \textit{preventing short sellers from accelerating a declining
market by exhausting all remaining bids at one price level, causing
successively lower prices to be established by long sellers.}
\end{itemize}

The last two objectives have been partially confirmed by the empirical analysis (see, e.g., \citep{AlexanderPeterson1999} and reference therein).
Instead, the regulation does not seem to be effective in producing the first desired effect.
The observed number of executed short sales is indeed lower under uptick rule than in the unconstrained case, during phases with an upward market trend, see again \citep{AlexanderPeterson1999}.
This is due to the asynchrony between placement and execution of a short-sell order, since the rising of the price in between these two operations can make the trade not feasible under the uptick rule.

Moreover, empirical evidence provides uniform support of
the idea that short selling restrictions often cause share prices to rise.
From a theoretical point of view, there is no clear argument for explaining
this mispricing effect of the uptick rule. According to \citep{Miller1977}
this is due to a reduction in stock supply owing to the short sale restriction.
More generally, theoretical models with heterogeneous agents and differences in trading strategies
support the idea that share values become overvalued under short selling
restrictions due to the fact that ``pessimistic'' and ``bear'' traders
(expecting negative price movements) are ruled out of the market
(see, e.g., \citep{HarrisonKreps1978}). In contrast, theoretical models based on the
assumption that all agents have rational expectations suggest that short
selling restrictions do not change on the average the stock prices
(see, e.g., \citep{DiamondVerecchia1987}).

However, given the complexity of the phenomena, and the impossibility of
isolating the effects of a regulation from other concomitant changes in the
economic scenario, the effectiveness of the uptick rule in meeting the three
above objectives, and its possible side effects on shares' prices,
are still far from being completely clarified.

Guided by the aim to provide further insight on the argument, this paper
studies the effects on share prices in an artificial market of a short
selling restriction based on a tick test similar to the one imposed by the
uptick rule in real financial markets. Using an artificial asset pricing
model makes it easier in assessing the effects of the uptick rule in
isolation from other exogenous shocks, though artificial modeling
necessarily trades realism for mathematical tractability.

We consider an asset pricing model of adaptive rational equilibrium dynamics (A.R.E.D.), where heterogeneous beliefs on the future prices of a risky
asset, together with traders' adaptability based on past performances, have shown
to endogenously sustain price fluctuations.
Asset pricing models of A.R.E.D. (hereafter referred to simply as ARED asset pricing models) are discrete-time dynamical systems 
based on the empirical evidence that investors with different trading strategies coexist in the financial market
(see, e.g., \citep{TaylorAllen1992}). These simple models provide a theoretical justification for many ``stylized facts'' observed in
the real financial time series, such as, financial bubbles and volatility clustering (see \citep{Gaunersdorfer2001}, and \citep{GaunersdorferHommesWagener2008}).
Stochastic models based on the same assumptions are even used to study
exchange rate volatility and the implication of some specific financial policies
(see, e.g., \citep{Westerhoff2001}).

We extend, in particular, the deterministic model introduced in \citep{BrockHommes1998},
where, in the simplest case, agents choose between two predictors of future prices of a risky asset, i.e. a fundamental predictor and a non-fundamental predictor. Agents that adopt the fundamental predictor are called {\it fundametalists}, while agents that adopt the non-fundamental predictor are called {\it noise traders} or non-fundamental traders.
Fundamentalists believe that the price of a financial asset is determined by its fundamental value (as given by the present discounted value
of the stream of future dividends, see \citep{Hommes2001}) and any deviation
from this value is only temporary.
Non-fundamental traders, sometimes called {\it chartists} or {\it technical analysts}, believe that the future price of a risky asset is not completely determined by fundamentals and it can be predicted by simple technical trading rules (see, e.g., \citep{Elder1993}, \citep{Murphy1999}, and \citep{Neely1997}).


In the model, agents revise their "beliefs", prediction to be adopted, according to an evolutionary mechanism based on the past realized profits.
As a result, the fundamental value is a fixed point of the price dynamics,
as, once there, both fundamentalists and non-fundamental traders predict
the fundamental price. As long as the sensitivity of traders in switching
to the best performing predictor is relatively low, the fundamental
equilibrium is stable, but the fundamental stability is typically lost at higher intensities of
the traders' choice across the predictors, making room for financial bubbles.

It is worth to remember that \citep{BrockHommes1998} investigated the peculiar case of zero supply of
outside shares. Under this assumption each bought share is sold short. We therefore
consider a positive supply of outside share, that is essential to ensure financial
transactions when short selling is forbidden%
\footnote{%
We consider a positive supply of outside shares for the asset pricing model under Walrasian market clearing at each period.
A similar model under the market maker scenario has been considered by \citep{HommesHuangWang2005}.}.
Moreover, we pair the fundamental predictor with first a technical linear predictor and then with a technical nonlinear
predictor and compare the results obtained with and without the uptick rule.

As linear predictor, we consider the chartist predictor introduced in \citep{BrockHommes1998}. This facilitates the comparison of our results with those in \citep{BrockHommes1998} and related papers.
As nonlinear predictor, we introduce a new predictor, "Smoothed Price Rate Of Change" or S-ROC predictor, that extrapolates future prices by applying the rate of change averaged on past prices with a confidence mechanism smoothing out extreme unrealistic rates (for an overview of this class of predictors, see \citep{Elder1993}).


For what concerns the implementation of the regulation, we implement the uptick rule as it was in its original formulation, i.e., short
selling is allowed only on an uptick. Note, however, that in an artificial
asset pricing model a zero-tick is possible only at equilibrium, so that
allowing or forbidding short sales on zero-plus-ticks makes basically no
change in the observed price dynamics. In fact, with a positive supply of
shares, traders take long positions at the fundamental equilibrium,
so only the non-fundamental equilibria at which one type of trader is prohibited to go
short are affected by the rule behavior on zero-plus-ticks (moreover, such equilibria
are irrelevant to study the global price dynamics, as will be explained in
Section \ref{ssec:ARED_uptick}).

From the mathematical point of view, the uptick rule makes the asset pricing
model a piecewise-smooth dynamical system\footnote{To be precise, the model is a piecewise-continuous dynamical system. However, the class of piecewise-smooth dynamical systems contains the class of piecewise-continuous dynamical systems.}, namely a system in which
different mathematical rules can be applied to determine the next price, and
the rule to be applied depends on the current state of the system, that is,
on the fact that trader types are interested in going short and whether short selling is allowed or not. Non-smooth dynamical systems are certainly more
problematic to analyze, both analytically and numerically 
(though non-smooth dynamics is a very active topic in current research, see \citep{Bernardo08}, and \citep{Colombo11b}, and
references therein)
so we will limit the analytical treatment to stationary solutions.

Piecewise-smooth dynamical systems have already been used as models in
finance. \citep{TramontanaGardiniWesterhoff2011} proposed a one-dimensional
piecewise-linear asset pricing model, where traders adopt different buying and selling
strategies in response to different market movements.
Other examples can be found in \citep{TramontanaGardiniWesterhoff2010}, \citep{TramontanaWesterhoffGardini2010}, and
\citep{TramontanaWesterhoff2012}.
Two ARED piecewise-smooth systems modeling short selling restrictions have been also proposed. Modifying the model in \citep{BrockHommes1998},
\citep{AnufrievTuinstra2009} restricted short selling by allowing limited short positions at each
trading period, whereas \citep{DercoleCecchetto2010} investigated the
complete ban on short selling. Thus both contributions implement short
selling restrictions that are not based on price tests.

The results of our theoretical analyses are in line with the empirical
evidence. The sale price that is established in our model when one trader
type is prohibited from going short is indeed systematically higher than the
unconstrained price. Thus, constrained downward movements below the
fundamental value are less pronounced, whereas constrained upward movements
above the fundamental value can be larger.
We provide a more complete explanation for this effect, suggesting that it is due to the combination of two mechanisms:
on one side, the short selling restriction reduces the possibility for pessimistic or bear traders to bet on downward movements
below the fundamental value,
avoiding excessive underpricing, but at the same time, when prices are above the fundamental value,
the restriction reduces the possibility for fundamentalists to drive down the prices back to the fundamental value by opening short positions.
This is in agreement with the last two goals established by the SEC (see above).
The first stated objective of the uptick is always realized in our model, since the market clearing is assumed to be synchronous among all traders. 

When non-fundamental traders adopt the S-ROC predictor, we observe that the overpricing due to the uptick rule disappears due to the smoothness of the predictor that makes non-fundamental traders not confident with extreme price deviations from the fundamental value. Indeed, the expectations of large price deviations produced by the uptick rule force the non-fundamental trader to believe in the fundamental value with the effect of reducing, instead of increasing, the price deviations.
The stabilizing effect however vanishes when traders become highly sensitive in switching to the strategy with best recent performance.

The paper is organized as follows. Section~\ref{ssec:ARED_no_uptick} briefly
reviews the unconstrained asset pricing model, summarizing from %
\citep{BrockHommes1998} and setting the notation and most of the modeling
equations that will be used in next Sections. Section~\ref{ssec:fund} is
also preliminary and recaps the concept of fundamental equilibrium,
including its stability analysis and some new results. Section~\ref%
{ssec:ARED_uptick} formulates the piecewise-smooth model constrained by the
uptick rule, and discusses the existence and stability of fundamental and
non-fundamental equilibria. So far, no explicit price predictors is
introduced, whereas Section \ref{ssec:pre} presents the price predictors
for which the unconstrained and constrained models will be studied and
compared in Sections~\ref{sec:ana} and~\ref{sec:num}. Section~\ref{sec:ana}
presents the analytical results concerning the existence and stability of
fixed points. Some of the results concerning the unconstrained model are new
and interesting per se. 
Section~\ref{sec:num} presents a series of numerical tests, confirming the
analytical results and investigating non-stationary (periodic, quasi-periodic, and chaotic)
regimes. In Section~\ref{sec:ed} we discuss in detail our economic findings.
Section~\ref{sec:cd} concludes and lists a series of related interesting topics for further research.
All the analytical results presented in Sections~\ref{sec:ARED} and~\ref{sec:ana}
are proved in Appendix \ref{Ch4Appendix}.

\section{The ARED asset pricing model with and without the uptick rule}\label{sec:om}
\label{sec:ARED}
We consider the asset pricing model with heterogeneous beliefs and adaptive traders introduced by \citep{BrockHommes1998}.
While in the original model a zero supply of outside shares was considered, making short selling essential to ensure the exchanges,
we consider the case of positive supply, so that short selling will no longer be necessary and a constraint on it can be imposed.
In this generalized version of the original model, we introduce a negative demand constraint according to the uptick rule,
in order to study the effects of this regulation on price fluctuations.

\subsection{The unconstrained ARED asset pricing model}
\label{ssec:ARED_no_uptick}
Consider a financial market where traders invest either in a
single risky asset, supplied in $S$ shares \footnote{$S$ is in fact the supply of traded assets in each period. Obviously when short selling is allowed assets are borrowed outside the pool of this $S$ shares making the total supply higher than $S$.}
of (ex-dividend) price $p_{t}$ at period $t$, or in a risk free asset perfectly elastically supplied at gross return $R$ (where $R=1+r$, with $r\in(0,1)$). The risky asset pays random dividend $\tilde{y}_{t}$ in period $t$, where the divided process $\tilde{y}_{t}$ is IID (Identically Independently Distributed) with $E_{t}\left[\tilde{y}_{t+1}\right]=\bar{y}$ constant. Thus, denoting by
$W_{h,t}$ the economic wealth of a generic trader of type $h$ at the beginning of period
$t$, and by $z_{h,t}$ the number of shares held by the trader in period $t$, we
have the following wealth equation (or individual intertemporal budget constraint):
\begin{linenomath}
\begin{equation}
\label{eq:w}
\tilde{W}_{h,t+1}=R(W_{h,t}-p_t z_{h,t})+\tilde{p}_{t+1}z_{h,t}+\tilde{y}_{t+1}z_{h,t}=
RW_{h,t}+(\tilde{p}_{t+1}+\tilde{y}_{t+1}-R\hspace{0.2mm}p_t)z_{h,t},
\end{equation}
\end{linenomath}
where tilde denotes random variables, $W_{h,t}-p_t z_{h,t}$ is the amount of money invested in the risk free asset in period $t$
and $\tilde{R}_{t+1}=\tilde{p}_{t+1}+\tilde{y}_{t+1}-R\hspace{0.2mm}p_t$ is the excess return per share realized at the end of the period.

Let $E_{h,t},V_{h,t}$ denote the "beliefs" of investor of type $h$ about the conditional expectation and conditional variance of wealth. They are assumed to be functions of past prices and dividends. We assume that each investor type is a myopic mean variance maximizer, so for type $h$ the demand for shares $z_{h,t}$ solves
\begin{equation*}
\max_{z_{h,t}}\left\{E_{h,t}\left(\tilde{W}_{t+1}\right)-\frac{a}{2}V_{h,t}\left(\tilde{W}_{t+1}\right)\right\}
\end{equation*}
i.e.,
\begin{linenomath}
$$
\ds z_{h,t}\left(p_{t}\right)=\Frac{E_{h,t}[\tilde{R}_{t+1}]}{a V_{h,t}[\tilde{R}_{t+1}]}=
\Frac{E_{h,t}[\tilde{p}_{t+1}+\tilde{y}_{t+1}]-R\hspace{0.2mm}p_t}{a V_{h,t}[\tilde{R}_{t+1}]},
$$
\end{linenomath}
where $a$ is the risk aversion coefficient and $p_t$ is to be determined by the market clearing between all demands and the supply $S$ of shares.
For simplicity (as done in \citep{BrockHommes1998}, see \citep{Gaunersdorfer00}, for an extension),
we assume that traders have common and constant beliefs about the variance, i.e. $V_{h,t}\left[\tilde{R}_{t+1}\right]=\sigma^{2}$, $\forall h$, and common and correct beliefs about the dividend, i.e. $E_{h,t}\left[\tilde{y}_{t+1}\right]=E_{t}\left[\tilde{y}_{t+1}\right]=\bar{y}$, $\forall h$.
Moreover, the number $N$ of traders and $S$ of supplied shares in each period (not considering the extra supply of shares due to short sales) are kept constant. Let $H$ be the number of available "beliefs" or price predictors $E_{h,t}[\tilde{p}_{t+1}+\tilde{y}_{t+1}]$, $h=1,\ldots,H$, each obtained at a cost $C_h$, and denote by $n_{h,t}$ the fraction of traders adopting predictor $h$ in
period $t$, the market clearing imposes
\begin{linenomath}
\begin{equation}
\label{eq:mc}
\ds N\sum_{h=1}^Hn_{h,t}z_{h,t}(p_t)=S,\quad
z_{h,t}(p_t)=\frac{\ds E_{h,t}[\tilde{p}_{t+1}+\tilde{y}_{t+1}]-R\hspace{0.2mm}p_t}{\ds a\hspace{0.2mm}\sigma^2},
\end{equation}
\end{linenomath}
which is solved for $p_t$, thus obtaining
\begin{linenomath}
\begin{equation}
\label{eq:pt}
p_{t}=\Frac{1}{R}\left(\sum_{h=1}^{H}n_{h,t}E_{h,t}[\tilde{p}_{t+1}+\tilde{y}_{t+1}]-a\hspace{0.2mm}\sigma^{2}\Frac{S}{N}\right),
\end{equation}
\end{linenomath}
Let us substitute the expression of $p_{t}$ in $z_{h,t}\left(p_{t}\right)$, $\forall h \in H$, to obtain the actual demands
\begin{linenomath}
\begin{equation}
\label{eq:ad}
z_{h,t}=\Frac{1}{a\hspace{0.2mm}\sigma^2}\left(E_{h,t}[\tilde{p}_{t+1}+\tilde{y}_{t+1}]-\sum_{k=1}^{H}n_{k,t}E_{k,t}[\tilde{p}_{t+1}+\tilde{y}_{t+1}]\right)+\Frac{S}{N}
\end{equation}
\end{linenomath}
(no longer functions of the price $p_t$), and let us use $p_{t}$ to calculate the net profits $R_t z_{h,t-1}-C_h$, $h=1,\ldots,H$, realized in period $t$.

Eq.~\eqref{eq:ad} gives the number of shares held by a trader of type $h$ in period $t$.
If negative, the trader is in a short position. If positive, the trader is in a long position.

At this point, the fractions $n_{h,t+1}$, $h=1,...,H$, for the next period are determined as functions of the positions of the traders and of the last available net profits.
In particular, the following discrete choice model is used:
\begin{linenomath}
\begin{equation}
\label{eq:cm}
n_{h,t+1}=\frac{\exp\left(\beta(R_t z_{h,t-1}-C_h)\right)}
{\sum_{k=1}^H\exp\left(\beta(R_t z_{k,t-1}-C_k)\right)}, h=1,...,H-1,
\end{equation}
\end{linenomath}
where $\beta$ measures the intensity of traders' choice across predictors (traders' adaptability).
The above procedure can then be iterated to compute the next price $p_{t+1}$.

If all agents have common beliefs on the future prices, i.e. $E_{h,t}=E_{t}$ $\forall h$, the pricing equation (\ref{eq:pt}) reduces to
\begin{equation*}
Rp_{t}=E_{t}[\tilde{p}_{t+1}+\tilde{y}_{t+1}]-a\hspace{0.2mm}\sigma^{2}\Frac{S}{N}.
\end{equation*} 
This equation admits a unique solution $\tilde{p}^{*}_{t}\equiv \bar{p}$, where
\begin{linenomath}
\begin{equation}
\label{eq:pbar}
\bar{p}=\Frac{\bar{y}-a\sigma^2 S/N}{R-1},
\end{equation}
\end{linenomath}
that satisfies the "no bubbles" condition $\lim_{t\rightarrow\infty}\left(E\tilde{p}^{*}_{t}/R^{t}\right)=0$. This price, given as the discounted sum of expected future dividends, would prevail in a perfectly rational world and is called the {\it fundamental price} (see, e.g., \citep{Hommes2001,HommesHuangWang2005}).
Of course, we assume $\bar{p}>0$, i.e., sufficiently high dividend $\bar{y}$ or limited supply of ouside shares per investor $S/N$.
Given the assumptions about the dividend process and the fundamental price and focusing only on the deterministic skeleton of the model, i.e. $\tilde{y}_{t}=\bar{y}$ $\forall t$, we have that $E_{h,t}[\tilde{p}_{t+1}+\tilde{y}_{t+1}]=E_h[p_{t+1}]+\bar{y}$, where the price predictors $E_h[p_{t+1}]$, $h=1,\ldots,H$,
are deterministic functions of $L$ known past prices $\{p_{t-1},p_{t-2},\dots,p_{t-L}\}$, $L\ge 1$.

It is useful to rewrite the model in terms of price deviations from a benchmark price $\bar{p}$. In the following, let $s=S/N$ and denote by $x_t$ the price deviation from the fundamental value, i.e., $x_t=p_t-\bar{p}$.
Defining the traders' beliefs on the next deviation $x_{t+1}$ as $f_h(\mathbf{x}_{t})=E_h[p_{t+1}]-\bar{p}$,
with $\mathbf{x}_{t}=(x_{t-1},x_{t-2},\dots,x_{t-L})$ being the vector of the last $L$ available deviations,
the demand functions to be used in the market clearing in Eq.~\eqref{eq:mc} become
\begin{linenomath}
\begin{equation}
\label{eq:zx}
z_{h,t}(x_t)=\Frac{f_{h}(\mathbf{x}_{t})-R\hspace{0.2mm}x_t}{a\hspace{0.2mm}\sigma^2}+s,
\end{equation}
\end{linenomath}
while the pricing equation \eqref{eq:pt} and the actual demands \eqref{eq:ad} can be written in deviations as
\begin{linenomath}
\begin{equation}
\label{eq:xt}
x_{t}=\Frac{1}{R}\sum_{h=1}^{H}n_{h,t}\hspace{0.2mm}f_{h}(\mathbf{x}_{t})\quad\text{and}\quad
z_{h,t}=\Frac{1}{a\hspace{0.2mm}\sigma^2}\left(f_{h}(\mathbf{x}_{t})-\sum_{k=1}^{H}n_{k,t}f_{k}(\mathbf{x}_{t})\right)+s,
\end{equation}
\end{linenomath}
and the excess of return in \eqref{eq:cm} can be expressed in deviations as
\begin{linenomath}
\begin{equation}
\label{eq:Rx}
R_t=x_t-R\hspace{0.2mm}x_{t-1}+\delta_{t}+a\hspace{0.2mm}\sigma^2s.
\end{equation}
\end{linenomath}
where $\delta_{t}=y_{t}-\bar{y}$ is a shock due to the dividend realization. As mentioned above, we focus on the deterministic skeleton of the model, i.e., we fix $\delta_{t}=0$ $\forall t$.


Substituting Eqs.~\eqref{eq:zx} and~\eqref{eq:Rx} into~\eqref{eq:cm} and coupling the pricing equation in \eqref{eq:xt} with \eqref{eq:cm}, the ARED model can be rewritten as
\begin{linenomath}
\begin{subequations}
\label{ARED_model_no_uptick}
\footnotesize
\begin{eqnarray}
\label{ARED_model_no_uptick_x}
x_{t} & = & \Frac{1}{R}\sum_{h=1}^{H}n_{h,t}\hspace{0.2mm}f_{h}(\mathbf{x}_{t}),\\
\label{ARED_model_no_uptick_n}
n_{h,t+1} & = & \Frac{\exp\left(\beta\left((x_{t}-R\hspace{0.2mm}x_{t-1}+a\hspace{0.2mm}\sigma^2s)\left(\Frac{f_{h}(\mathbf{x}_{t-1})-R\hspace{0.2mm}x_{t-1}}{a\hspace{0.2mm}\sigma^2}+s\right)-C_h\right)\right)}
{\sum_{k=1}^H\exp\left(\beta\left((x_{t}-R\hspace{0.2mm}x_{t-1}+a\hspace{0.2mm}\sigma^2s)\left(\Frac{f_{k}(\mathbf{x}_{t-1})-R\hspace{0.2mm}x_{t-1}}{a\hspace{0.2mm}\sigma^2}+s\right)-C_k\right)\right)},\;h=1,\ldots,H-1\nonumber\\
\end{eqnarray}
\end{subequations}
\end{linenomath}
(recall that $\sum_{h=1}^Hn_{h,t}=1$).
Given the current composition $n_{h,t}$ of the traders' population, the first equation computes the price deviation for period $t$, while the second updates the traders' fractions for the next period.
The past deviations $(x_{t-1},x_{t-2},\dots,x_{t-(L+1)})$ appearing in vectors $\mathbf{x}_{t}$ and $\mathbf{x}_{t-1}$, together with the fractions $n_{h,t}$, $h=1,\ldots,H-1$, constitute the state of the system\footnote{%
Note that, by writing Eq.~\eqref{ARED_model_no_uptick_n} for $n_{h,t}$ and substituting it into Eq.~\eqref{ARED_model_no_uptick_x}, one can write $x_t$ as a recursion on the last $L+2$ deviations.
This gives a more compact and homogeneous system's state ($L+2$ price deviations instead of $L+1$ deviations and $H-1$ traders' fractions),
however, the formulation \eqref{ARED_model_no_uptick} is physically more appropriate and easier to initialize.}.

The initial condition is composed of the opening price deviation $x_0$ and of the traders' fractions $n_{h,1}$, $h=1,...,H$, to be used in the first period.
In fact, assuming that each price predictor can be customized to the case when the number of past available prices is less then $L$,
then Eq.~\eqref{ARED_model_no_uptick_x} can be applied at $t=1$ (to determine the price deviation $x_1$ in period $1$), whereas Eq.s~\eqref{ARED_model_no_uptick_n} can only be applied at $t=2$, so that $n_{h,2}=n_{h,1}$ is used.
For $t>L$ the price predictors in Eqs.~\eqref{ARED_model_no_uptick} can be regularly applied.

Note that Eq.~\eqref{ARED_model_no_uptick_x} guarantees a positive price for any period $t$ (i.e., $x_t>-\bar{p}$),
provided all price predictions are such ($f_{h}(\mathbf{x}_{t})>-\bar{p}$ for all $h=1,\ldots,H$).

When there are only two types of traders, $H=2$, it is convenient to express the
fractions $n_{1,t}$ and $n_{2,t}$ as a function of $m_{t}=n_{1,t}-n_{2,t}\in(-1,1)$, i.e.,
\begin{linenomath}
\begin{equation}
\label{eq:m}
n_{1,t}=\Frac{1+m_{t}}{2}\quad\text{and}\quad n_{2,t}=\Frac{1-m_{t}}{2}.
\end{equation}
\end{linenomath}
In this specific case, model \eqref{ARED_model_no_uptick} can be rewritten as
\begin{linenomath}
\begin{subequations}
\label{ARED_model_no_uptick_2}
\begin{eqnarray}
\label{ARED_model_no_uptick_2x}
x_{t} & = & \Frac{1}{2R}\left((1+m_{t})f_1(\mathbf{x}_{t})+(1-m_{t})f_2(\mathbf{x}_{t})\right),\\
\label{ARED_model_no_uptick_2m}
m_{t+1} & = & \tanh\left(\Frac{\beta}{2}\left((x_{t}-R\hspace{0.2mm}x_{t-1}+a\hspace{0.2mm}\sigma^{2}s)
\Frac{f_1(\mathbf{x}_{t-1})-f_2(\mathbf{x}_{t-1})}{a\hspace{0.2mm}\sigma^{2}}-(C_1-C_2)\right)\right),\nonumber\\
\end{eqnarray}
\end{subequations}
\end{linenomath}
where $(x_{t-1},x_{t-2},\dots,x_{t-(L+1)},m_t)$ is the system's state and $(x_0,m_1)$ identifies the initial condition.

\subsection{The fundamental equilibrium}
\label{ssec:fund}
The following lemma gives the condition under which the fundamental price is an equilibrium of model \eqref{ARED_model_no_uptick} (or model \eqref{ARED_model_no_uptick_2} when $H=2$):
\begin{lemma}
\label{lm:fund}
If all predictors satisfy $f_h(\mathbf{0})=0$, $h=1,\ldots,H$, with $\mathbf{0}$ the vector of $L$ zeros, then
$(\bar{x}^{(0)},\bar{n}^{(0)}_h)$ with
\begin{linenomath}
$$
\bar{x}^{(0)}=0\quad\text{and}\quad
\bar{n}^{(0)}_h=\frac{\exp\left(-\hspace{-0.2mm}\beta C_h)\right)}{\sum_{k=1}^H \exp\left(-\hspace{-0.2mm}\beta C_k)\right)}
$$
\end{linenomath}
[or $(\bar{x}^{(0)},\bar{m}^{(0)})$ with $\bar{m}^{(0)}=\tanh\left(-\beta/2\,(C_1-C_2)\right)$ if $H=2$]
is a fixed point of model \eqref{ARED_model_no_uptick} [\eqref{ARED_model_no_uptick_2}], at which all strategies equally demand $\bar{z}^{(0)}_h=s$.
We call this steady state \emph{fundamental equilibrium}.
$H$ of the associated eigenvalues are zero and the remaining $L$ ones are the roots of the characteristic equation
\begin{linenomath}
$$
\lambda^{L}-\gamma_{1}\lambda^{L-1}+\cdots-\gamma_{L}=0,\quad
\gamma_{i}=\Frac{1}{R}\sum_{h=1}^H\bar{n}_h\hspace{-1.0mm}\left.\Frac{\partial}{\partial x_{t-i}}f_h(\mathbf{x}_{t})\right|_{\mathbf{x}_{t}=\mathbf{0}},\quad i=1,\ldots,L.
$$
\end{linenomath}
\end{lemma}

Lemma~\ref{lm:fund} also reveals that the price dynamics is not reversible, at least locally to the fundamental equilibrium (due to the presence of zero eigenvalues), so that prices cannot be reconstructed backward in time.

We now state a simple condition that rules out the possibility of other equilibria:
\begin{lemma}
\label{lm:neq}
If $f_h(\bar{x}\mathbf{1})/\bar{x}<R$ [or if $f_h(\bar{x}\mathbf{1})/\bar{x}>R$] for all $h=1,\ldots,H$ and $\bar{x}\neq 0$, with $\mathbf{1}$ the vector of $L$ ones,
then the fundamental equilibrium is the only fixed point of model \eqref{ARED_model_no_uptick}.
\end{lemma}

As we will recall in Section~\ref{ssec:pre}, traders with price predictors such that $|f_h(\bar{x}\mathbf{1})/\bar{x}|<1$ believe that tomorrow's price will revert to its fundamental value ($x_t\to 0$), whereas, at an equilibrium, trend followers obviously extrapolate the equilibrium price, so their price predictors are such that $f_h(\bar{x}\mathbf{1})/\bar{x}=1$.
Lemma \ref{lm:neq} therefore shows that non-fundamental equilibria are possible only in the presence of at least one of the two mentioned types of traders and traders that believe that nonzero price deviations will amplify in the short run, even if they have been recently constant.


\subsection{The ARED asset pricing model constrained by the uptick rule}
\label{ssec:ARED_uptick}
When trading restrictions imposed by the uptick-rule are introduced, we must distinguish between two situations:
if prices are rising, e.g. we simply look at the last movement available $x_{t-1}-x_{t-2}$,
then short selling is allowed and the unconstrained model \eqref{ARED_model_no_uptick} still applies.
In contrast, in a downward (or stationary) movement ($x_{t-1}\le x_{t-2}$), traders' demands are forced to be non-negative, i.e.,
the demand functions to be used in the market clearing in Eq.~\eqref{eq:mc} are
\begin{linenomath}
\begin{equation}
\label{eq:zxc}
z_{h,t}(x_t)=\max\left\{0,
\Frac{f_{h}(\mathbf{x}_{t})-R\hspace{0.2mm}x_t}{a\hspace{0.2mm}\sigma^2}+s\right\}.
\end{equation}
\end{linenomath}
Note that the forward dynamics remain uniquely defined.
In fact, given the past price deviations in $\mathbf{x}_{t}$ and the traders' fractions $n_{h,t}$,
the per capita demand $d(x_{t})=\sum_{h=1}^Hn_{h,t}z_{h,t}(x_{t})$ is a piecewise-linear, continuous function
of the deviation $x_{t}$, that is decreasing up to the deviation at which it vanishes together
with the highest of the single agents' demand curves, and $d(x_{t})=0$ for larger deviations (see Figure~(\ref{fig:mc})).
There is therefore a unique deviation $x_t$ at which the market clears, i.e., $d(x_t)=s>0$.
Also note that, given the same traders' fractions, the constrained price\footnote{In this part of the paper we introduce the notation $x_{U,t}$ for the unconstrained price determined by model in Section \ref{ssec:ARED_no_uptick} to distinguish it from the constrained price $x_{t}$. Since there is no risk of confusion, this distinction is not made in other parts of the paper for the sake of avoiding cumbersome notations.} is higher than the unconstrained price\footnote{It is clear that all the traders' demand functions are always characterized by the same slope. However, the intercept of the demands with the $x=0$ axis, i.e. $s+\frac{f_{h}(\mathbf{x}_{t})}{a\hspace{0.2mm}\sigma^2}$, changes over time and it depends on the past share prices. For example if the predictor of a trader is based on $L$ past prices of the share, i.e. $\mathbf{x}_{t}=\left(x_{t-1},...,x_{t-L}\right)$, the intercept of its demand with the $x=0$ axis depends on all of these prices. Thus, to prove that the constrained prices are always higher than the unconstrained one given the same past prices, it does not necessarily mean a price dynamic characterized by larger fluctuations for the constrained model. The situation can be the opposite when we consider predictors based on a large number of past deviations and especially when they are non-linear. In other words, simple static considerations on the shape of the constrained demands do not help us to understand entirely the effect of the uptick rule on the price dynamics.} ($x_{t}>x_{U,t}$),
so that positive predictions ($f_{h}(\mathbf{x}_{t})>-\bar{p}$ for all $h=1,\ldots,H$) still yield
positive prices ($x_t>-\bar{p}$).

\begin{figure}[t]
\centerline{\includegraphics[scale=0.9]{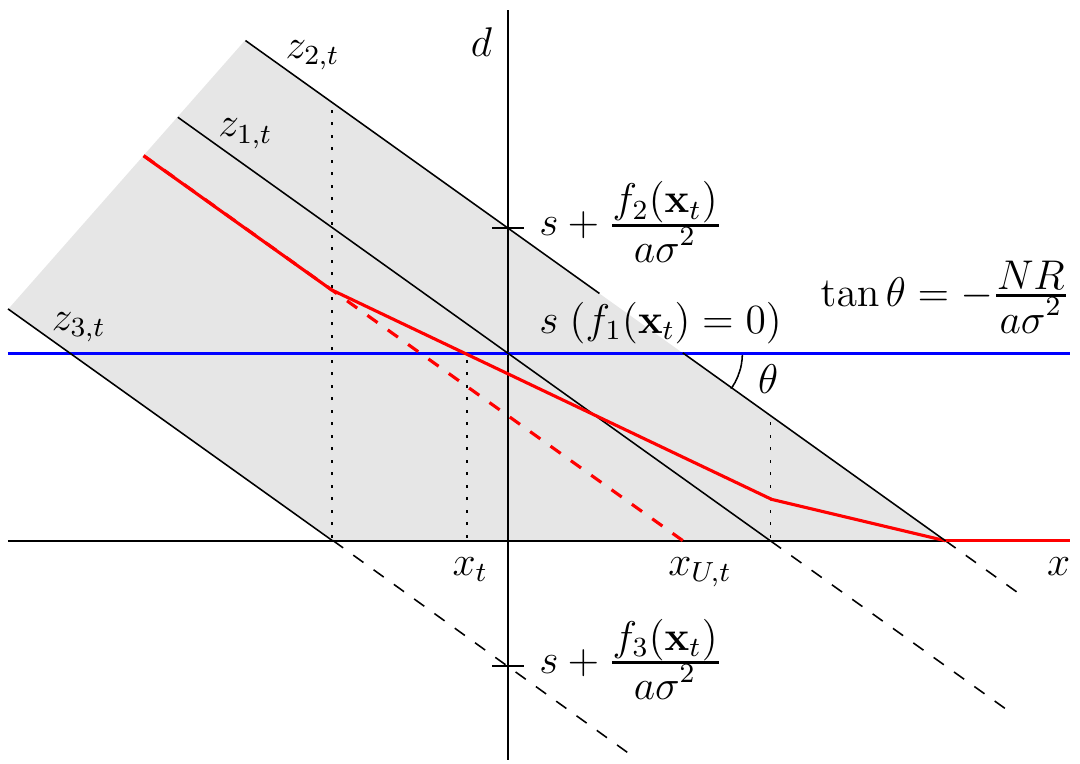}}
\caption{Per-capita demand ($d$, red) and supply ($s$, blue) curves as functions
of the price deviation $x$ to be realized in period $t$.
A case with $H=3$ types of traders is sketched, where $z_{1,t}$ is the demand
of "fundamentalists" (see Sect.~\ref{ssec:pre}), while $z_{2,t}$ and $z_{3,t}$ are the demands of "non-fundamentalists" ($f_{2}(\mathbf{x}_t)>0$, $f_{3}(\mathbf{x}_t)<0$, and
negative (dashed) demands are obtained with \eqref{eq:zx}).
The resulting per capita demand curve is piecewise-linear, continuous, and decreasing.
Depending on the traders' fractions ($n_{1,t}$, $n_{2,t}$, $n_{3,t}$), it can take on different
configurations in the shaded area
(the most negative of which corresponds to $z_{3,t}$ when $n_{3,t}=1$;
the case shown corresponds to $n_{1,t}=n_{2,t}=n_{3,t}=1/3$).
At the unconstrained price $x_{U,t}$ traders of type $3$ are in a short position in period $t$.}
\label{fig:mc}
\end{figure}

When solving Eq.~\eqref{eq:mc} for $x_t$ with the constrained demands \eqref{eq:zxc}, $2^H-1$ cases
must be further distinguished, depending on which of the optimal demands in \eqref{eq:zx} are forced to zero by \eqref{eq:zxc}
(obviously not all demands can vanish).
The uniqueness of forward dynamics guarantees that only one of the cases clears the market.

For simplicity, hereafter we will only consider the case with two types of traders ($H=2$), so
one of the following three cases is realized at each period:
\begin{itemize}
\item[$0$:] no trader is prohibited from going short (equivalently, both types of traders hold nonnegative amounts of shares in period $t$), i.e.,\\
\eqref{eq:zx} implies $z_{1,t}\ge 0$ and $z_{2,t}\ge 0$, with
$x_t=\Frac{1}{R}\left(n_{1,t}\,f_{1}(\mathbf{x}_{t})+n_{2,t}\,f_{2}(\mathbf{x}_{t})\right)$,
\item[$1$:] traders of type $1$ are prohibited from going short (only traders of type $2$ hold shares in period $t$), i.e.,\\
\eqref{eq:zx} implies $z_{1,t}<0$ and $z_{2,t}>0$, with
$x_{t}=\Frac{1}{R}\left(f_{2}(\mathbf{x}_{t})-a\hspace{0.2mm}\sigma^{2}s\Frac{n_{1,t}}{n_{2,t}}\right)$,
\item[$2$:] traders of type $2$ are prohibited from going short (only traders of type $1$ hold shares in period $t$), i.e.,\\
\eqref{eq:zx} implies $z_{1,t}>0$ and $z_{2,t}<0$, with
$x_{t}=\Frac{1}{R}\left(f_{1}(\mathbf{x}_{t})-a\hspace{0.2mm}\sigma^{2}s\Frac{n_{2,t}}{n_{1,t}}\right)$.
\end{itemize}

Then, $m_{t+1}$ must be computed (see \eqref{eq:m}) and, again, there are three cases, depending
on the signs of the optimal demands at period $(t-1)$ (see Eq.~\eqref{eq:cm}).
In order to simplify the model formulation, we prefer to enlarge the system's state, by including the traders' demands $z_{h,t-1}$ realized in period $(t-1)$, $h=1,2$, in lieu of the farthest price deviation $x_{t-(L+1)}$.
The state variables therefore are
\begin{subequations}
\begin{linenomath}
\begin{equation}
\label{eq:state}
\ensuremath\makebox[50mm][r]{$(x_{t-1}, x_{t-2},\dots, x_{t-L},\ z_{1,t-1},\ z_{2,t-1},\ m_t)$}\quad\text{if}\quad L\ge 2,
\end{equation}
\begin{equation}
\label{eq:stateL1}
\ensuremath\makebox[50mm][r]{$(x_{t-1}, x_{t-2},\ z_{1,t-1},\ z_{2,t-1},\ m_t)$}\quad\text{if}\quad L=1,
\end{equation}
\end{linenomath}
\end{subequations}
as we need $x_{t-2}$ to apply the uptick rule,
and we can update the traders' fractions by simply replacing \eqref{ARED_model_no_uptick_2m} with  
\begin{linenomath}
\begin{equation}
\label{eq:mz}
m_{t+1} = \tanh\left(\Frac{\beta}{2}\left((x_{t}-R\hspace{0.2mm}x_{t-1}+a\hspace{0.2mm}\sigma^{2}s)
(z_{1,t-1}-z_{2,t-1})-(C_1-C_2)\right)\right).
\end{equation}
\end{linenomath}

The uptick rule makes the ARED model piecewise smooth, namely the space of the state variables
is partitioned into three regions associated with different equations for updating the system's state (see \citep{Bernardo08} and references therein).
By defining the regions
\begin{subequations}
\label{eq:reg}
\begin{linenomath}
\begin{equation}
\label{eq:regZ}
\begin{array}{lrrr}
\multicolumn{1}{c}{U}: & x_{t-1}>x_{t-2},\\
Z_{0}: & x_{t-1}\le x_{t-2}, & z_{1,t}\ge 0, & z_{2,t}\ge 0,\\
Z_{1}: & x_{t-1}\le x_{t-2}, & z_{1,t}<0, & z_{2,t}>0,\\
Z_{2}: & x_{t-1}\le x_{t-2}, & z_{1,t}>0, & z_{2,t}<0,

\end{array}
\end{equation}
\end{linenomath}
where
\begin{linenomath}
\begin{equation}
\label{eq:regz}
z_{1,t} = \Frac{1-m_{t}}{2}\,\Frac{f_{1}(\mathbf{x}_{t})-f_{2}(\mathbf{x}_{t})}{a\hspace{0.2mm}\sigma^2}+s\quad\text{and}\quad
z_{2,t} = \Frac{1+m_{t}}{2}\,\Frac{f_{2}(\mathbf{x}_{t})-f_{1}(\mathbf{x}_{t})}{a\hspace{0.2mm}\sigma^2}+s
\end{equation}
\end{linenomath}
\end{subequations}
are the optimal demands from \eqref{eq:xt}, we can write the forward dynamics as follows:
\begin{linenomath}
\begin{subequations}
\label{ARED_model_uptick}
\begin{eqnarray}
\label{ARED_model_uptick_x}
x_t & = & \left\{\begin{array}{ll}
\ensuremath\makebox[60mm][l]{$\Frac{1}{2R}\left((1+m_{t})f_1(\mathbf{x}_{t})+(1-m_{t})f_2(\mathbf{x}_{t})\right)$} &
\quad\text{if}\quad (\mathbf{x}_{t},m_t)\in U \cup Z_{0},\\[2mm]
\Frac{1}{R}\left(f_{2}(\mathbf{x}_{t})-a\hspace{0.2mm}\sigma^{2}s\hspace{0.2mm}\Frac{1+m_t}{1-m_t}\right) &
\quad\text{if}\quad (\mathbf{x}_{t},m_t)\in Z_{1},\\[2mm]
\Frac{1}{R}\left(f_{1}(\mathbf{x}_{t})-a\hspace{0.2mm}\sigma^{2}s\hspace{0.2mm}\Frac{1-m_t}{1+m_t}\right) &
\quad\text{if}\quad (\mathbf{x}_{t},m_t)\in Z_{2},\\[2mm]
\end{array}\right.\\
\label{ARED_model_uptick_z1}
z_{1,t} & = & \left\{\begin{array}{ll}
\ensuremath\makebox[60mm][l]{$\Frac{1-m_{t}}{2}\,\Frac{f_{1}(\mathbf{x}_{t})-f_{2}(\mathbf{x}_{t})}{a\hspace{0.2mm}\sigma^2}+s$} &
\quad\text{if}\quad (\mathbf{x}_{t},m_t)\in U \cup Z_{0},\\[2mm]
0 &
\quad\text{if}\quad (\mathbf{x}_{t},m_t)\in Z_{1},\\[2mm]
\Frac{2\hspace{0.2mm}s}{1+m_t} &
\quad\text{if}\quad (\mathbf{x}_{t},m_t)\in Z_{2},\\[2mm]
\end{array}\right.\\
\label{ARED_model_uptick_z2}
z_{2,t} & = & \left\{\begin{array}{ll}
\ensuremath\makebox[60mm][l]{$\Frac{1+m_{t}}{2}\,\Frac{f_{2}(\mathbf{x}_{t})-f_{1}(\mathbf{x}_{t})}{a\hspace{0.2mm}\sigma^2}+s$} &
\quad\text{if}\quad (\mathbf{x}_{t},m_t)\in U \cup Z_{0},\\[2mm]
\Frac{2\hspace{0.2mm}s}{1-m_t} &
\quad\text{if}\quad (\mathbf{x}_{t},m_t)\in Z_{1},\\[2mm]
0 &
\quad\text{if}\quad (\mathbf{x}_{t},m_t)\in Z_{2},\\[2mm]
\end{array}\right.\\
\label{ARED_model_uptick_m}
m_{t+1} & = & \tanh\left(\Frac{\beta}{2}\left((x_{t}-R\hspace{0.2mm}x_{t-1}+a\hspace{0.2mm}\sigma^{2}s)
(z_{1,t-1}-z_{2,t-1})-(C_1-C_2)\right)\right).
\end{eqnarray}
\end{subequations}
\end{linenomath}

The same model can be rewritten in compact notations as follows:


\begin{linenomath}
\begin{eqnarray}\label{ARED_model_uptick_CF}
\nonumber
\left(x_t,z_{1,t},z_{2,t},m_{t+1}\right) & = & \left\{\begin{array}{ll}
\ensuremath\makebox[10mm][l]{$G_{1}\left(\mathbf{x}_{t},z_{1,t-1},z_{2,t-1},m_{t}\right)$} &
\quad\text{if}\quad (\mathbf{x}_{t},m_t)\in U \cup Z_{0},\\[2mm]
G_{2}\left(\mathbf{x}_{t},z_{1,t-1},z_{2,t-1},m_{t}\right) &
\quad\text{if}\quad (\mathbf{x}_{t},m_t)\in Z_{1},\\[2mm]
G_{3}\left(\mathbf{x}_{t},z_{1,t-1},z_{2,t-1},m_{t}\right) &
\quad\text{if}\quad (\mathbf{x}_{t},m_t)\in Z_{2},\\[2mm]
\end{array}\right.\\
\end{eqnarray}
\end{linenomath}

where $G_1$, $G_2$ and $G_3$ are three systems that define the asset pricing model with uptick rule.


\begin{figure}[t]
\centerline{\includegraphics[scale=0.7]{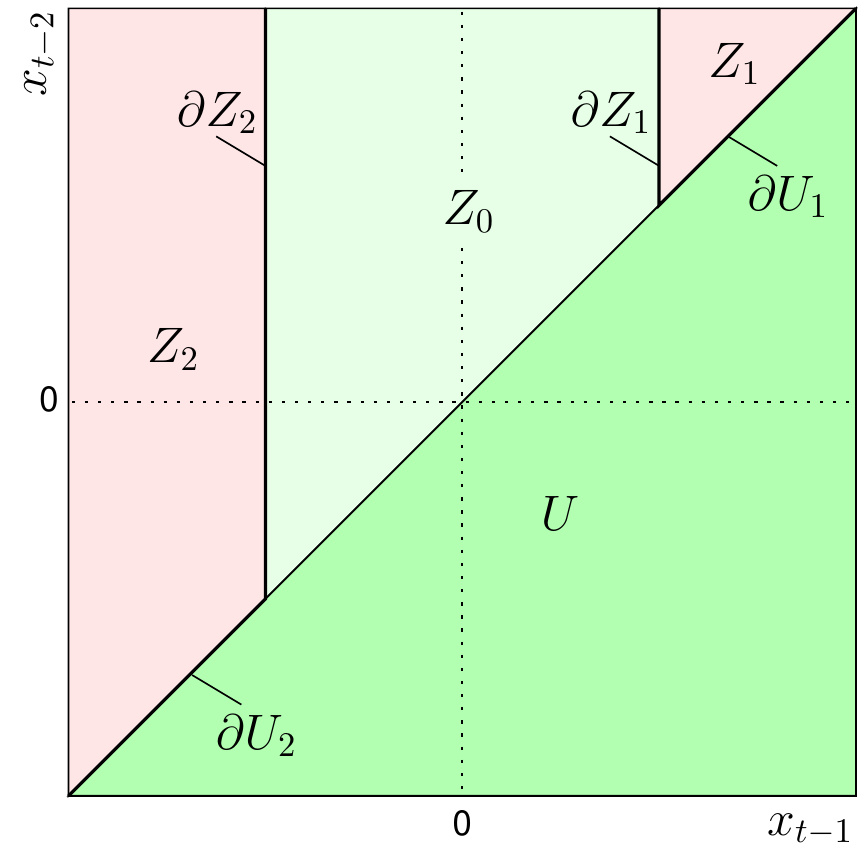}}
\caption{Partition of the state space into the three regions $U \cup Z_{0}$ (green) and $Z_h$, $h=1,2$ (pink):
the projection on the $(x_{t-1},x_{t-2})$ space in the case where traders use the fundamental and chartist predictors (see Section~\ref{ssec:pre}).}
\label{fig:ns}
\end{figure}

Region $U \cup Z_{0}$ is separated from region $Z_h$, $h=1,2$, by the two boundaries
\begin{linenomath}
\begin{equation}
\label{eq:bor}
\partial U_h:\; x_{t-1}=x_{t-2},\; z_{h,t}\le 0\quad\text{and}\quad
\partial Z_h:\; x_{t-1}\le x_{t-2},\; z_{h,t}=0
\end{equation}
\end{linenomath}
(see Figure~(\ref{fig:ns}), where a projection on the $(x_{t-1},x_{t-2})$ space is shown).
Across boundary $\partial U_k$ the system is discontinuous, i.e., the corresponding expressions on the right-hand sides of
(\ref{ARED_model_uptick_x}--c) assume different values on $\partial U_k$.
In contrast, the system is continuous (but not differentiable) at the boundaries $\partial Z_h$.

Similarly to the unconstrained model \eqref{ARED_model_no_uptick}, the initial condition of model \eqref{ARED_model_uptick} is set by the opening price deviation $x_0$ and the traders' initial composition $m_1$. However, Eq.~\eqref{ARED_model_uptick_m} also requires the traders' demand $z_{1,0}$ and $z_{2,0}$ which can be conventionally set at $s$.

Finally, let's discuss the fixed points of model \eqref{ARED_model_uptick} (or equivalently \eqref{ARED_model_uptick_CF}), which necessarily lie on the boundary of region $U$ and are still denoted with the pair $(\bar{x},\bar{m})$
(the equilibrium demands, which also characterize the equilibria of model \eqref{ARED_model_uptick}, can be obtained from Eqs.~(\ref{ARED_model_uptick_z1},c)).
Each of the three systems defining model \eqref{ARED_model_uptick}, system $G_{1}$ (equivalent to model \eqref{ARED_model_no_uptick} adding the demands as state variables) which defines the dynamics of the model in region $U \cup Z_{0}$ and the two systems, $G_{1}$ and $G_{2}$, which define the dynamics of the model in regions $Z_1$ and $Z_2$, respectively, have their own fixed points, which we call either {\it admissible} or {\it virtual} according to the region of the state space in which they are.
In this paper, a generic fixed point or equilibrium of system $G_{1}$ is called admissible if it lies in region $U \cup Z_{0}$ and virtual otherwise, a generic fixed point of system $G_{2}$ is called admissible if it lies in region $Z_{1}$ and virtual otherwise, and a generic fixed point of system $G_{3}$ is called admissible if it lies in region $Z_{2}$ and virtual otherwise.
Virtual fixed points are not equilibria of model \eqref{ARED_model_uptick}, but tracking their position is useful in the analysis.


The fundamental equilibrium is always an admissible fixed point of system $G_{1}$, also called the unconstrained system because equivalent to model \eqref{ARED_model_no_uptick} adding the demands as state variables. Indeed, it lies on the boundary between regions $U$ and $Z_{0}$, with positive demands (equal to $s$), i.e., it is always an interior point of region $U \cup Z_{0}$.
Its local stability is therefore ruled by Lemma~\ref{lm:fund}
(the storage of the previous demands in lieu of the farthest past deviation brings the number of zero eigenvalues to $2H-1$, $2H$ if $L=1$, see \eqref{eq:stateL1}),
while the existence of (admissible or virtual) non-fundamental equilibria of the unconstrained system $G_{1}$ is ruled by Lemma~\ref{lm:neq}.

The fixed points of the other two systems, $G_2$ and $G_3$, are of little interest.
They lie on the boundaries $\partial U_2$ and $\partial U_1$, respectively, across which model \eqref{ARED_model_uptick} is discontinuous.
Hence, there are arbitrarily small perturbations from the fixed point entering region $U$, for which the unconstrained system $G_1$ will map the system's state far from the fixed point.
The fixed points of the two systems $G_2$ and $G_{3}$ are therefore (highly) unstable and will not be considered in the analysis.

\subsection{Classical price predictors}
\label{ssec:pre}
In this Section we briefly introduce the price predictors used in this paper (see, e.g., \citep{Elder1993}, \citep{ChiarellaHe2002}, and \citep{ChiarellaHe2003}, for a more complete survey of the most classical types of price predictors used in the literature).
The first one, $f_1(\mathbf{x}_{t})$, called fundamental predictor, will be paired with each of the others, the non-fundamental predictors, in the analysis of Sections~\ref{sec:ana} and~\ref{sec:num}.
For this reason, all non-fundamental predictors will be denoted by $f_2(\mathbf{x}_{t})$.

\subsubsection*{Fundamental predictor}
Fundamental traders, or {\it fundamentalists}, believe that prices return to their fundamental value.
The simplest fundamental prediction is the fundamental price for period $t+1$, irrespectively of the recent trend:
\begin{linenomath}
\begin{equation}
\label{eq:pr1}
E_1[p_{t+1}]=\bar{p},\quad f_1(\mathbf{x}_{t})=0.
\end{equation}
\end{linenomath}
More generally, fundamentalists believe that prices will revert to the fundamental value by a factor $v$ at each period:
\begin{linenomath}
\begin{equation}
\label{eq:pr1v}
E_1[p_{t+1}]=\bar{p}+v\hspace{0.2mm}(p_{t-1}-\bar{p}),\quad f_1(\mathbf{x}_{t})=v\hspace{0.2mm}x_{t-1},\quad 0\le v<1,
\end{equation}
\end{linenomath}
the smaller is $v$, the highest is the expected speed of convergence to the fundamental price.


As in \citep{BrockHommes1998}, we assume that "training" costs must be borne to obtain enough "understanding" of how markets work in order to believe in the fundamental price, so fundamentalists incur into a cost $C_{1}>0$ at each prediction.

\subsubsection*{Chartist predictor}
The second type of simple trader that we consider is called {\it chartist} or {\it trend chaser}. This type of trader believes that any mispricing will continue, i.e. the chartist predictor is formally equivalent to predictor \eqref{eq:pr1v}:
\begin{linenomath}
\begin{equation}
\label{eq:pr1g}
E_2[p_{t+1}]=\bar{p}+g\hspace{0.2mm}(p_{t-1}-\bar{p}),\quad f_2(\mathbf{x}_{t})=g\hspace{0.2mm}x_{t-1},\quad g>1,
\end{equation}
\end{linenomath}
but amplifies, instead of damping, nonzero price deviations from the fundamental.



The chartist prediction is not costly.



\subsubsection*{Rate of change (ROC) predictor}

The third type of simple trader that we consider is called "{\it nonlinear technical analyst}" or "{\it ROC trader}".

The ROC ({\it"Price Rate Of Change"}) is a nonlinear prediction which applies the price rate of change averaged over the last $L-1$ periods,
\begin{subequations}
\begin{linenomath}
\begin{equation}
\label{eq:roc}
\text{ROC} = \left(\Frac{p_{t-1}}{p_{t-L}}\right)^{\hspace{-1.0mm}\frac{1}{L-1}}=\left(\Frac{\bar{p}+x_{t-1}}{\bar{p}+x_{t-L}}\right)^{\hspace{-1.0mm}\frac{1}{L-1}},\quad L\ge 2,
\end{equation}
\end{linenomath}
to $p_{t-1}$ twice to extrapolate $p_{t+1}$:
\begin{linenomath}
\begin{equation}
\label{eq:rocp}
E_2[p_{t+1}] = p_{t-1}\hspace{0.2mm}\text{ROC}^{2},\quad
f_2(\mathbf{x}_{t}) = (\bar{p}+x_{t-1})\hspace{0.2mm}\text{ROC}^{2}-\bar{p}.
\end{equation}
\end{linenomath}
The ROC predictor is typically "smoothed'' to avoid extreme rates of change (rates that are either too high or too close to zero, see Smoothed-ROC or S-ROC predictors in \citep{Elder1993}).
This interprets the traders' rationality that makes them diffident with extreme rates.
We adopt in particular the confidence mechanism introduced in \citep{DercoleCecchetto2010},
where the ROC is combined with the last available price.
Precisely, the price rate of change to be applied is a convex combination of the actual ROC \eqref{eq:roc} and the unitary rate (corresponding to the last available price),
with the ROC weight $\alpha_{\text{ROC}}$ that vanishes when the ROC attains extreme values (zero and infinity):
\begin{linenomath}
\begin{eqnarray}
\label{eq:rocpa}
E_2[p_{t+1}] & = & p_{t-1}\left(\alpha_{\text{ROC}}\hspace{0.2mm}\text{ROC}+(1-\alpha_{\text{ROC}})\right)^{2},\nonumber\\
f_2(\mathbf{x}_{t}) & = & (\bar{p}+x_{t-1})\left(\alpha_{\text{ROC}}\hspace{0.2mm}\text{ROC}+(1-\alpha_{\text{ROC}})\right)^{2}-\bar{p}.
\end{eqnarray}
\end{linenomath}
The function
\begin{linenomath}
\begin{equation}
\label{eq:roca}
\alpha_{\text{ROC}}=\Frac{2}{\text{ROC}^{\hspace{0.5mm}\alpha}+\text{ROC}^{\hspace{0.5mm}-\alpha}}
\end{equation}
\end{linenomath}
\end{subequations}
has been used in the analysis, where the parameter $1/\alpha$ measures how confident traders are with extreme rates.

The ROC predictor and the S-ROC predictor are not costly.  


\section{The effect of the uptick rule on shares price fluctuations: Analytical results}
\label{sec:ana}
In this Section we report the stability analysis of fundamental and non-fundamental equilibria
of models \eqref{ARED_model_no_uptick_2} and \eqref{ARED_model_uptick} for two pairs of traders' types.
As traditionally done in the literature, type $1$ is always the fundamental type (price predictor \eqref{eq:pr1v}),
while type two is either the chartist in Sect.~\ref{ssec:cha_a} (predictor \eqref{eq:pr1g})
or the nonlinear technical analyst (ROC trader) in Sect.~\ref{ssec:roc_a} (predictor (\ref{eq:roc},c,d)). 


\subsection{Fundamentalists vs chartists}
\label{ssec:cha_a}
Consider models \eqref{ARED_model_no_uptick_2} and \eqref{ARED_model_uptick} with predictors \eqref{eq:pr1v} and \eqref{eq:pr1g}.
Model (\ref{ARED_model_no_uptick_2},~\ref{eq:pr1v},~\ref{eq:pr1g}) is the classical ARED model, proposed and fully analyzed in \citep{BrockHommes1998} for the case of zero supply of outside shares, i.e., $s=0$, where short selling is intrinsically practiced at each trading period.
The case with positive supply is analyzed in \citep{AnufrievTuinstra2009}, where the effects of a negative bound on the traders' positions are also investigated.

Without any constraint on short selling, 
the existence and stability of the fixed points of model (\ref{ARED_model_no_uptick_2},~\ref{eq:pr1v},~\ref{eq:pr1g}) are defined in the following lemma:

\begin{lemma}
\label{lm:cha}
The following statements hold true for the dynamical system \eqref{ARED_model_no_uptick_2} with predictors \eqref{eq:pr1v} and \eqref{eq:pr1g}:
\begin{enumerate}
\item
For $1<g<R$ the fundamental equilibrium $(0,\bar{m}^{(0)})$ (see Lemma~\ref{lm:fund}) is the only fixed point and is globally stable.
\item
For $R<g<2R-v$ there are the following possibilities:
\begin{enumerate}
\item[\rm (a)]
For $0\hspace{-0.2mm}\le\hspace{-0.2mm}\beta\hspace{-0.2mm}<\hspace{-0.2mm}\beta_{\mathrm{LP}}=\hspace{-0.2mm}\Frac{1}{C}\log\hspace{-0.3mm}\left(\Frac{R\hspace{-0.2mm}-\hspace{-0.2mm}v}{g\hspace{-0.2mm}-\hspace{-0.3mm}R}\right)\hspace{-0.7mm}\left(1\hspace{-0.3mm}+\hspace{-0.2mm}\Frac{a\hspace{0.2mm}\sigma^2s^2}{4C}\Frac{g\hspace{-0.2mm}-\hspace{-0.2mm}v}{R\hspace{-0.2mm}-\hspace{-0.2mm}1}\right)^{\hspace{-1mm}-1}\hspace{-1.0mm}>\hspace{-0.3mm}0$
the fundamental equilibrium is the only fixed point and is stable.
\item[\rm (LP)]
At $\beta\hspace{-0.1mm}=\hspace{-0.1mm}\beta_{\mathrm{LP}}$ two equilibria appear (as $\beta$ increases) at
\begin{linenomath}
$$
\bar{x}_{\mathrm{LP}}\hspace{-0.3mm}=\hspace{-0.2mm}\Frac{a\hspace{0.2mm}\sigma^2s}{2(R\hspace{-0.3mm}-\hspace{-0.4mm}1)}\hspace{-0.3mm}>\hspace{-0.3mm}0,\quad
\bar{m}=1-2\hspace{0.2mm}\Frac{R-v}{g-v}, 
$$
\end{linenomath}
through a saddle-node bifurcation (limit point, \emph{LP}).
\item[\rm (b)]
For $\beta_{\mathrm{LP}}\hspace{-0.2mm}<\hspace{-0.2mm}\beta\hspace{-0.2mm}<\hspace{-0.2mm}\min\left\{\beta_{\mathrm{TR}}=\hspace{-0.2mm}\Frac{1}{C}\log\hspace{-0.3mm}\left(\Frac{R\hspace{-0.2mm}-\hspace{-0.2mm}v}{g\hspace{-0.2mm}-\hspace{-0.3mm}R}\right),\beta_{\mathrm{NS}}^{(+)}\right\}>\beta_{\mathrm{LP}}$
the fundamental equilibrium is locally (asymptotically) stable and coexists with the two non-fundamental equilibria $(\bar{x}^{(\pm)},\bar{m})$, with
\begin{linenomath}
$$
\bar{x}^{(\pm)}=\bar{x}_{\mathrm{LP}}\pm\sqrt{\left(\bar{x}_{\mathrm{LP}}^2+\Frac{a\hspace{0.2mm}\sigma^2C}{(R\hspace{-0.3mm}-\hspace{-0.4mm}1)(g\hspace{-0.2mm}-\hspace{-0.2mm}v)}\right)\hspace{-0.5mm}\left(\hspace{-0.5mm}1\hspace{-0.4mm}-\hspace{-0.3mm}\Frac{\beta_{\mathrm{LP}}}{\beta}\hspace{-0.2mm}\right)\hspace{-0.3mm}}\hspace{0.3mm}.
$$
\end{linenomath}
Equilibrium $(\bar{x}^{(+)},\bar{m})$ is locally (asymptotically) stable, whereas $(\bar{x}^{(-)},\bar{m})$ is a saddle with 2-dimensional stable manifold separating the basins of attraction of the two stable equilibria.
\item[\rm (TR)]
At $\beta\hspace{-0.1mm}=\hspace{-0.1mm}\beta_{\mathrm{TR}}$, $(\bar{x}^{(-)},\bar{m})$ collides and exchanges stability with the fundamental equilibrium (transcritical bifurcation, \emph{TR}). The fundamental equilibrium is always at least locally (asymptotically) stable for $\beta\hspace{-0.1mm}<\hspace{-0.1mm}\beta_{\mathrm{TR}}$ and it is always unstable for $\beta\hspace{-0.1mm}>\hspace{-0.1mm}\beta_{\mathrm{TR}}$.
\item[\rm (NS$^{(+)}$)]
At $\beta\hspace{-0.1mm}=\hspace{-0.1mm}\beta_{\mathrm{NS}}^{(+)}$ the equilibrium $(\bar{x}^{(+)},\bar{m})$ undergoes a Neimark-Sacker (NS) bifurcation.
No explicit expression is available for $\beta_{\mathrm{NS}}^{(+)}$, but $\beta_{\mathrm{TR}}\lessgtr\beta_{\mathrm{NS}}^{(+)}$ if $\beta_{\mathrm{TR}}\lessgtr\Frac{1}{a\hspace{0.2mm}\sigma^2s^2}\Frac{(R\hspace{-0.2mm}-\hspace{-0.2mm}1)^2}{(g\hspace{-0.2mm}-\hspace{-0.2mm}R)(R\hspace{-0.2mm}-\hspace{-0.2mm}v)}$.
\item[\rm (c)]
For $\beta_{\mathrm{TR}}\hspace{-0.2mm}<\hspace{-0.2mm}\beta\hspace{-0.2mm}<\hspace{-0.2mm}\beta_{\mathrm{NS}}^{(-)}>\beta_{\mathrm{TR}}$
the fundamental equilibrium is a saddle, with 2-dimensional stable manifold separating the positive from the negative dynamics,
and the equilibrium $(\bar{x}^{(-)},\bar{m})$, with $\bar{x}^{(-)}\hspace{-0.4mm}<0$, is stable.
\item[\rm (NS$^{(-)}$)]
At $\beta\hspace{-0.1mm}=\hspace{-0.1mm}\beta_{\mathrm{NS}}^{(-)}$ the equilibrium $(\bar{x}^{(-)},\bar{m})$ undergoes a Neimark-Sacker bifurcation.
\end{enumerate}
\item
For $g>2R-v$ there are the following possibilities:
\begin{enumerate}
\item[\rm (a)]
For $0\hspace{-0.2mm}\le\hspace{-0.2mm}\beta\hspace{-0.2mm}<\hspace{-0.2mm}\beta_{\mathrm{NS}}^{(\pm)}$
the fundamental equilibrium is unstable and the equilibria $(\bar{x}^{(\pm)},\bar{m})$ ($\bar{x}^{(+)}\hspace{-0.4mm}>0$ and $\bar{x}^{(-)}\hspace{-0.4mm}<0$) are stable.
\item[\rm (NS)]
At $\beta\hspace{-0.1mm}=\hspace{-0.1mm}\beta_{\mathrm{NS}}^{(\pm)}$ the equilibria $(\bar{x}^{(\pm)},\bar{m})$ undergo a Neimark-Sacker bifurcation.
\end{enumerate}
\item
For $g>R^2$ the dynamics can be unbounded for sufficiently large $\beta$. 
\end{enumerate}
\end{lemma}

Lemma \ref{lm:cha} generalizes Lemmas 2,~3, and~4 in \citep{BrockHommes1998} to the case $s>0$, $0<v<1$, and Proposition~3.1 in \citep{AnufrievTuinstra2009}.
In particular, for $s=0$, note that the saddle-node and transcritical bifurcations concomitantly occur (case $2$) at a so-called pitchfork bifurcation,
whereas the mechanism making the fundamental equilibrium unstable is different for $s>0$.
First, the two non-fundamental equilibria $(\bar{x}^{(\pm)},\bar{m})$ appear (as the traders' adaptability $\beta$ increases) through the saddle-node bifurcation,
and as $\beta$ increases further a transcritical bifurcation occurs in which the saddle $(\bar{x}^{(-)},\bar{m})$ exchanges stability with the fundamental equilibrium.
Thus, for $\beta_{\mathrm{LP}}<\beta<\beta_{\mathrm{TR}}$, the fundamental equilibrium is stable,
but coexists with an alternative stable fixed point of model (\ref{ARED_model_no_uptick_2},~\ref{eq:pr1v},~\ref{eq:pr1g}).




With the uptick-rule,
the existence and stability of the fixed points of model (\ref{ARED_model_uptick},~\ref{eq:pr1v},~\ref{eq:pr1g}) are complemented by the following lemma:

\begin{lemma}
\label{lm:cha_uptick}
The following statements hold true for the dynamical system \eqref{ARED_model_uptick} with predictors \eqref{eq:pr1v} and \eqref{eq:pr1g}:
\begin{enumerate}
\item
The local and global stability of the fundamental equilibrium is as in Lemma~\ref{lm:cha}.
As long as equilibria $(\bar{x}^{(\pm)},\bar{m})$ exist and are admissible, their local stability is as in Lemma~\ref{lm:cha}.
Equilibrium $(\bar{x}^{(+)},\bar{m})$ is admissible iff $\bar{x}^{(+)}\hspace{-0.4mm}\le\hspace{-0.1mm}\bar{x}_{\mathrm{BC}}^{(+)}\hspace{-0.2mm}=a\hspace{0.2mm}\sigma^2s/(R-v)$.
Equilibrium $(\bar{x}^{(-)},\bar{m})$ is admissible iff $\bar{x}_{\mathrm{BC}}^{(-)}\hspace{-0.2mm}=-a\hspace{0.2mm}\sigma^2s/(g\hspace{-0.2mm}-\hspace{-0.2mm}R)\le\bar{x}^{(-)}\hspace{-0.3mm}\le\bar{x}_{\mathrm{BC}}^{(+)}$.
\item
For $R<g<2R-v$ there are the following possibilities:
\begin{enumerate}
\item[\rm (a)]
If $R\hspace{-0.2mm}+\hspace{-0.2mm}v>2$, equilibria $(\bar{x}^{(\pm)},\bar{m})$ appear admissible at $\beta=\beta_{\mathrm{LP}}$ and $(\bar{x}^{(+)},\bar{m})$ becomes virtual (border-collision bifurcation) at $\beta=\beta_{\mathrm{BC}}^{(+)}$, with
\begin{linenomath}
$$
0<\beta_{\mathrm{BC}}^{(+)}=\Frac{1}{C}\log\hspace{-0.3mm}\left(\Frac{R\hspace{-0.2mm}-\hspace{-0.2mm}v}{g\hspace{-0.2mm}-\hspace{-0.3mm}R}\right)\hspace{-0.7mm}\left(1\hspace{-0.3mm}+\hspace{-0.2mm}\Frac{a\hspace{0.2mm}\sigma^2s^2}{C}\Frac{(g\hspace{-0.2mm}-\hspace{-0.2mm}v)(1\hspace{-0.2mm}-\hspace{-0.2mm}v)}{(R\hspace{-0.2mm}-\hspace{-0.2mm}v)^2}\right)^{\hspace{-1mm}-1}\hspace{-1.0mm}<\beta_{\mathrm{TR}}.
$$
\end{linenomath}
\item[\rm (b)]
If $R\hspace{-0.2mm}+\hspace{-0.2mm}v<2$, equilibria $(\bar{x}^{(\pm)},\bar{m})$ appear virtual at $\beta=\beta_{\mathrm{LP}}$ and $(\bar{x}^{(-)},\bar{m})$ becomes admissible at $\beta=\beta_{\mathrm{BC}}^{(+)}$.
\item[\rm (c)]
If $R\hspace{-0.2mm}+\hspace{-0.2mm}v=2$, equilibria $(\bar{x}^{(\pm)},\bar{m})$ appear on the border $\partial Z_1$ at $\beta=\beta_{\mathrm{LP}}=\beta_{\mathrm{BC}}^{(+)}$ and $(\bar{x}^{(+)},\bar{m})$ and $(\bar{x}^{(-)},\bar{m})$ are respectively virtual and admissible for larger $\beta$.
\item[\rm (d)]
If $s<s_{\mathrm{BC}}^{(-)}=\left(\Frac{C}{a\hspace{0.2mm}\sigma^2}\Frac{(g-R)^2}{(g-v)(g-1)}\right)^{\hspace{-0.7mm}1\hspace{-0.2mm}/2}$,
equilibrium $(\bar{x}^{(-)},\bar{m})$ becomes virtual at $\beta=\beta_{\mathrm{BC}}^{(-)}$, with
\begin{linenomath}
$$
\beta_{\mathrm{BC}}^{(-)}=\Frac{1}{C}\log\hspace{-0.3mm}\left(\Frac{R\hspace{-0.2mm}-\hspace{-0.2mm}v}{g\hspace{-0.2mm}-\hspace{-0.3mm}R}\right)\hspace{-0.7mm}\left(1\hspace{-0.3mm}-\hspace{-0.2mm}\Frac{a\hspace{0.2mm}\sigma^2s^2}{C}\Frac{(g\hspace{-0.2mm}-\hspace{-0.2mm}v)(g\hspace{-0.2mm}-\hspace{-0.2mm}1)}{(g\hspace{-0.2mm}-\hspace{-0.2mm}R)^2}\right)^{\hspace{-1mm}-1}\hspace{-1.0mm}>\beta_{\mathrm{TR}}.
$$
\end{linenomath}
\end{enumerate}
\item
For $g>2R-v$ there are the following possibilities:
\begin{enumerate}
\item
If $s<s_{\mathrm{BC}}^{(-)}$ equilibria $(\bar{x}^{(\pm)},\bar{m})$ are virtual for any $\beta\ge0$.
\item
If $s>s_{\mathrm{BC}}^{(-)}$ equilibrium $(\bar{x}^{(+)},\bar{m})$ is virtual for any $\beta\ge0$, whereas $(\bar{x}^{(-)},\bar{m})$ is admissible for $\beta>\beta_{\mathrm{BC}}^{(-)}>0$.
\end{enumerate}
\item
For $g>R^2$ the dynamics can be unbounded for sufficiently large $\beta$. 
\end{enumerate}
\end{lemma}

Note that the uptick rule affects the price dynamics also when the supply of outside shares is large. Indeed, independently on $s$, there is always an equilibrium, $(\bar{x}^{(+)},\bar{m})$ or $(\bar{x}^{(-)},\bar{m})$, becoming virtual as $\beta$ increases or decreases.

\subsection{Fundamentalists vs ROC traders}
\label{ssec:roc_a}
Consider models \eqref{ARED_model_no_uptick_2} and \eqref{ARED_model_uptick} with the fundamental predictor \eqref{eq:pr1v} and with the ROC predictor (\ref{eq:roc},b) or (\ref{eq:roc},c,d).

Note that both predictors are such that $f_h(\bar{x}\mathbf{1})/\bar{x}<R$ for any possible equilibrium $(\bar{x},\bar{m})$ with $\bar{x}\neq0$, so by means of Lemma~\ref{lm:neq} the fundamental equilibrium is the only fixed point.
In the simplest case $L=2$, its stability is characterized in the following lemma:


\begin{lemma}
\label{lm:roc}
The following statements hold true for the dynamical systems \eqref{ARED_model_no_uptick_2} and \eqref{ARED_model_uptick} with predictors \eqref{eq:pr1v} and {\rm(}\ref{eq:roc},b{\rm)}, as well as with predictors \eqref{eq:pr1v} and {\rm(}\ref{eq:roc},c,d{\rm)}:
\begin{enumerate}
\item
For $R\ge 2$ the fundamental equilibrium $(0,\bar{m}^{(0)})$ (see Lemma~\ref{lm:fund}) is a stable fixed point for any $\beta>0$.
\item
For $R<2$ the fundamental equilibrium is stable for
\begin{linenomath}
$$
0<\beta<\beta_{\mathrm{NS}}=\Frac{1}{C}\log\hspace{-0.3mm}\left(\Frac{R}{2\hspace{-0.2mm}-\hspace{-0.2mm}R}\right)
$$
\end{linenomath}
and loses stability through a Neimark-Sacker bifurcation at $\beta=\beta_{\mathrm{NS}}$.
\end{enumerate}
\end{lemma}

Similarly to the cases where chartists
are paired with fundamentalists
(Sect.~\ref{ssec:cha_a}),
the stability of the fundamental equilibrium is guaranteed if the gross return $R$ is sufficiently large.

The stability analysis for $L>2$ is possible, following the lines indicated in \citep{Kuruklis1994}, and the general conclusion is that rates of change calculated on larger windows of past prices stabilize the fundamental equilibrium, up to the point that the fixed point is stable for any value of $\beta$ if $L$ is sufficiently large.

\section{The effect of the uptick rule on share price fluctuations: Numerical simulations.}
\label{sec:num}
In the first two Subsections of this Section we report several numerical analysis of models \eqref{ARED_model_no_uptick_2} and \eqref{ARED_model_uptick} for the two pairs of traders' types considered in Sect. \ref{sec:ana}, with the aim of characterizing the non-stationary (periodic, quasi-periodic, or chaotic) asymptotic regimes. This part contains technical considerations. In the last subsection, we discuss the effects of the uptick rule on the price dynamics.

As done in most of the related works in the literature, we use the traders' adaptability (or intensity of choice) $\beta$ as a bifurcation parameter (two- or higher-dimensional bifurcation analyses are possible, see, e.g., \citep{DercoleCecchetto2010}, but will not be considered here).
For each considered value of $\beta$, the transient dynamics is eliminated by computing the (largest) Lyapunov exponent associated to the orbit\footnote{%
The largest Lyapunov exponent is a measure of the mean divergence of nearby trajectories; it is positive, zero, and negative in chaotic, quasi-periodic, and periodic (or stationary) regimes, respectively \citep{Alligood96}.}, i.e., we delete the number of initial iterations required to compute the largest Lyapunov exponent,
whereas the asymptotic regime is discussed.
To graphically study the bifurcations underwent by the different attractors, we vertically plot the deviations $x_t$ in the attractor at the corresponding value of $\beta$, together with the associated largest Lyapunov exponent L (see, e.g., Figure~(\ref{fig:BH_s})).

In each simulation, we set the initial condition as follows.
The opening price deviation $x_0$ is randomly selected in a small, positive or negative neighborhood of zero to study the stability of the fundamental equilibrium; far from zero to study non-fundamental attractors.
The initial traders' fractions are equally set ($m_1=0$, i.e., $n_{1,1}=n_{2,1}=1/2$).
For model \eqref{ARED_model_uptick}, the initial values assigned to the traders' demands $z_{1,0}$ and $z_{2,0}$ are irrelevant, as the traders' fractions are not updated at $t=1$.

\subsection{Fundamentalists vs chartists}
\label{ssec:cha_n}
We first study the effects of a positive supply of outside shares ($s>0$) on the dynamics of the original model (\ref{ARED_model_no_uptick_2},~\ref{eq:pr1},~\ref{eq:pr1g}) introduced by \citep{BrockHommes1998}, then we study the effects of the uptick rule.
Figure~(\ref{fig:BH_s}) reports the bifurcation diagrams and the corresponding largest Lyapunov exponent obtained for four different values of $s$
(with $s$ in the range of values commonly used in the literature, see \citep{AnufrievTuinstra2009,HommesHuangWang2005}).
The first panel is the case with zero supply of outside shares ($s=0$) and is included for comparison.

\begin{figure}[t]
\centerline{\includegraphics[width=0.7\textwidth]{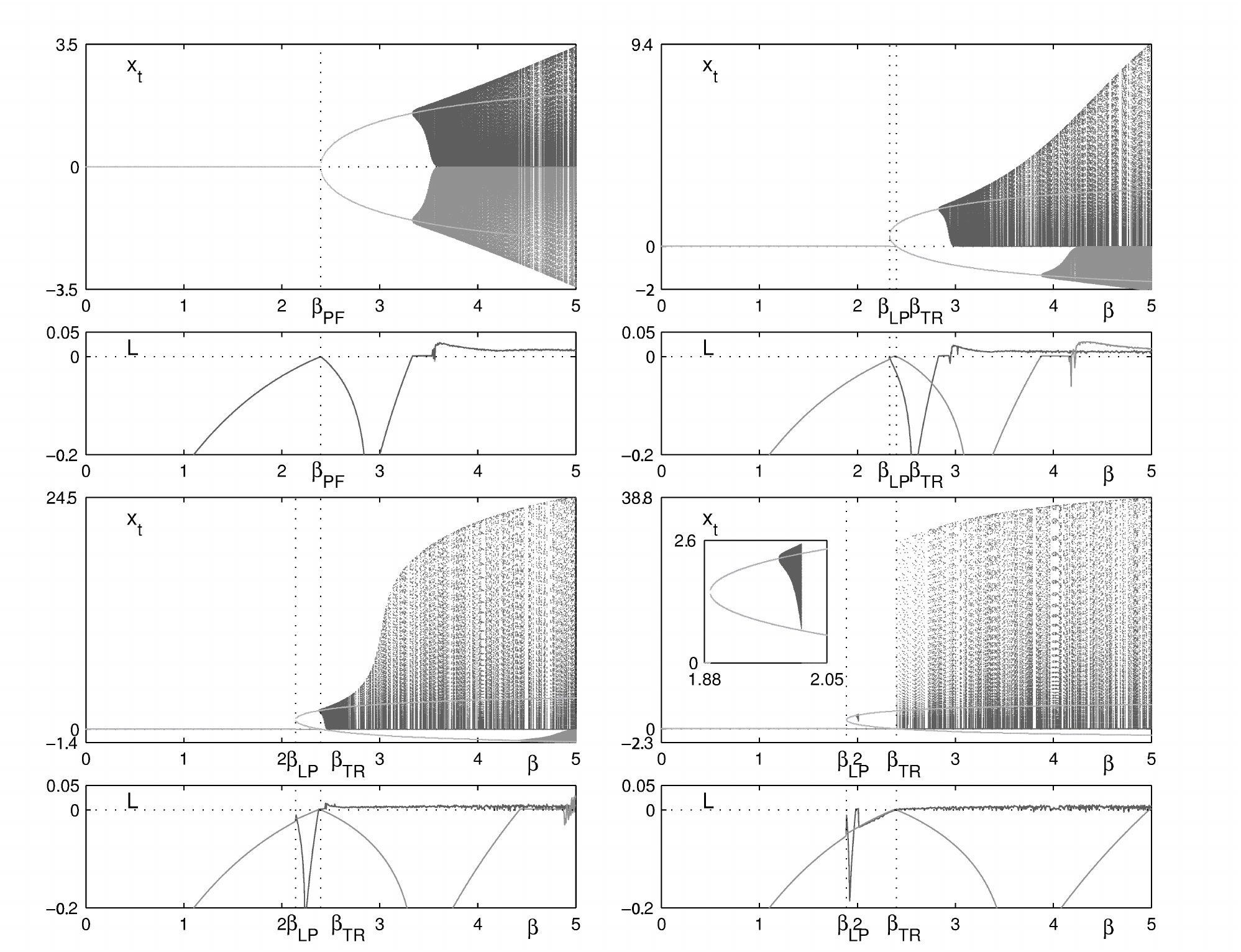}}
\caption{Bifurcation diagrams of model (\ref{ARED_model_no_uptick_2},~\ref{eq:pr1},~\ref{eq:pr1g}) for different values of $s$ (the supply of outside shares of the risky asset): A (first row, left column) $s=0$, in this case $\beta_{\mathrm{TR}}=\beta_{\mathrm{LP}}=\beta_{\mathrm{PF}}$, related (largest) Lyapunov exponents in (second row, left column); B (first row, right column) $s=0.1$, in this case $\beta_{\mathrm{TR}}<\beta_{\mathrm{NS}}^{(+)}$, related (largest) Lyapunov exponents in (second row, right column); C (third row, left column) $s=0.2$, in this case $\beta_{\mathrm{TR}}=\beta_{\mathrm{NS}}^{(+)}$, related (largest) Lyapunov exponents (fourth row, left column); D (third row, right column) $s=0.3$, in this case $\beta_{\mathrm{TR}}>\beta_{\mathrm{NS}}^{(+)}$, related (largest) Lyapunov exponents (fourth row, right column). $\beta_{\mathrm{NS}}^{(+)}$ and $\beta_{\mathrm{NS}}^{(-)}$ are not ticked because there is no analytical expression for them. However, the Neimark-Sacker bifurcation values can be clearly identified in the bifurcation diagrams.
Other parameter values: $R=1.1$, $a=1$, $\sigma=1$, $\bar{y}=1$, $g=1.2$, $C=1$.}
\label{fig:BH_s}
\end{figure}

The bifurcation points $\beta_{\mathrm{LP}}$, $\beta_{\mathrm{TR}}$, and $\beta_{\mathrm{NS}}$ are indicated, together with the non-fundamental equilibria $(\bar{x}^{(\pm)},\bar{m})$ (the gray parabola), and are in agreement with the analytical results of Lemma~\ref{lm:cha}.
In particular, $\beta_{\mathrm{LP}}=\beta_{\mathrm{TR}}=\beta_{\mathrm{PF}}$ indicates a pitchfork bifurcation when $s=0$.

Note that the positive and negative deviation dynamics are separated in model (\ref{ARED_model_no_uptick_2},~\ref{eq:pr1},~\ref{eq:pr1g}).
Indeed, due to the characteristics of the chartist predictor \eqref{eq:pr1g}, if the opening price is above [below] the fundamental value ($x_0>0$ [$x_0<0$]), it will remain so forever (see Eq.~\eqref{ARED_model_no_uptick_2x}).
The positive and negative attractors therefore coexist (with basins of attraction separated by the linear manifold $x_{t-1}=0$ in state space), so that two Lyapunov exponents (one for each of the two attractors) are plotted.

As expected, the fundamental equilibrium is destabilized for sufficiently high traders' adaptability and the amplitude of the price fluctuations increases as $\beta$ increases.
But the amplitude of fluctuations also increases with the supply of outside shares of the risky asset $s$ (note the different vertical scales in Figure~(\ref{fig:BH_s})).
The latter effect is partially due to the risk premium, i.e. $\left(a\sigma^{2}s\right)$, required by traders for holding extra shares that modifies the performance measures of the trading strategies.

Figure~(\ref{fig:BH_s}) shows another interesting dynamical phenomenon. For sufficiently high supply of outside shares of the risky asset ($s=0.3$ in case D), the positive attractor, that appears at the Neimark-Sacker (NS) bifurcation ($\beta_{\mathrm{NS}}^{(+)}$) of equilibrium $(\bar{x}^{(+)},\bar{m})$, collapses through a "homoclinic'' contact with the saddle equilibrium $(\bar{x}^{(-)},\bar{m})$.
The chaotic behavior is then reestablished when the fundamental equilibrium looses stability through the transcritical bifurcation ($\beta_{\mathrm{TR}}$).

The analysis of the dynamics conducted up to now reveals that there are not substantial differences in the price dynamics for different values ​​of $s$, except the amplitude of price fluctuations and some peculiar phenomena, such as the homoclinic contact just discussed. For this reason and to make the discussion more clear and concise, in the following we investigate the effects of the uptick rule on the price dynamics only for $s = 0.1$.

Figure~(\ref{fig:BHbif}) compares models \eqref{ARED_model_no_uptick_2} (left) and \eqref{ARED_model_uptick} (right) with predictors \eqref{eq:pr1} and \eqref{eq:pr1g}.
The bifurcation diagram and the largest Lyapunov exponents associated with the different attractors are reported, together with two examples of state portraits (projections in the plane $(x_{t},x_{t-1})$, see bottom panels).
Blue and red dots identify the points in the attractor in which respectively fundamentalists and chartists have been prohibited from going short.

\begin{figure}[t]
\centerline{\includegraphics[width=0.7\textwidth]{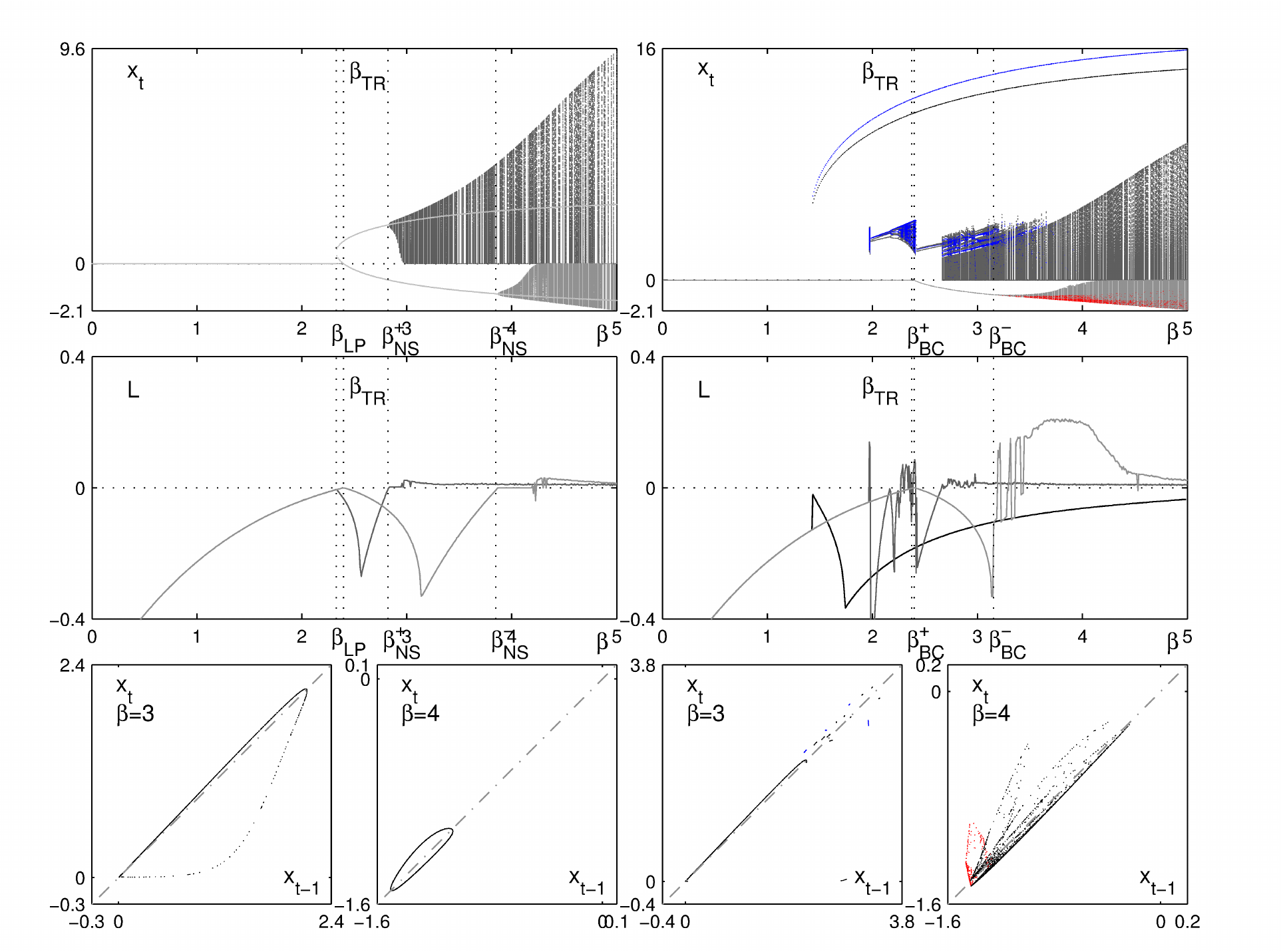}}
\caption{Bifurcation diagrams of models (\ref{ARED_model_no_uptick_2},~\ref{eq:pr1},~\ref{eq:pr1g}) (left) and (\ref{ARED_model_uptick},~\ref{eq:pr1},~\ref{eq:pr1g}) (right)
and exemplary state portraits projections in the plane $(x_{t},x_{t-1})$.
Parameter values as in Figure~(\ref{fig:BH_s})-case(B). The light gray represents the attractor of negative deviations and the corresponding largest Lyapunov exponent. The dark gray represents the attractor of positive deviations and the corresponding largest Lyapunov exponent. The black represents the attractor born through a nonsmooth saddle-node bifurcation and the corresponding largest Lyapunov exponent. The blue dots indicate that fundamentalists cannot have negative positions because of the uptick rule. The red dots indicate that chartists cannot have negative positions because of the uptick rule.}
\label{fig:BHbif}
\end{figure}


The first thing to remark is that multiple positive attractors are present in the constrained dynamics (right column in Figure~(\ref{fig:BHbif})), i.e. when the uptick rule is in place.
In particular a period-two cycle alternating unrestricted trading with restricted trading in which fundamentalists are forced by the uptick rule to take nonnegative positions (black and blue dots) is present for sufficiently large $\beta$. It appears through a nonsmooth saddle-node bifurcation and coexists with the chaotic attractor generated through the loss of fundamental stability.
Before the bifurcation the fundamental equilibrium is globally stable (while it is globally stable in the unconstrained dynamics up to the saddle-node at $\beta=\beta_{\mathrm{LP}}$).

Also the bifurcation structure leading to the chaotic attractor is more involved. The first branch of attractors suddenly appears as $\beta$ increases.
This is probably due to a homoclinic contact with the period-two saddle cycle, but this conjecture has not been verified.
This first branch seems to collapse through a collision with the border $\partial Z_1$, in connection with the border collision of the non-fundamental equilibrium $(\bar{x}^{(+)},\bar{m})$ at $\beta=\beta_{\mathrm{BC}}^{(+)}$.
The remaining thinner attractor later explodes into a larger one, again due to a border collision with $\partial Z_1$.
However, a deeper mathematical investigation would be required to confirm the above conjectures.

As for the negative equilibrium $(\bar{x}^{(-)},\bar{m})$ (for $\beta>\beta_{\mathrm{TR}}$), it loses stability at $\beta=\beta_{\mathrm{BC}}^{(-)}$ through a border collision with $\partial Z_2$, giving way to a chaotic attractor characterized by restrictions on chartists.

Figure~(\ref{fig:BHtsp}) shows an example of time series on the positive chaotic attractor obtained for $\beta=3$
(left column: unconstrained dynamics; right column: dynamics constrained by the uptick rule).
\begin{figure}[t]
\centerline{\includegraphics[width=0.7\textwidth]{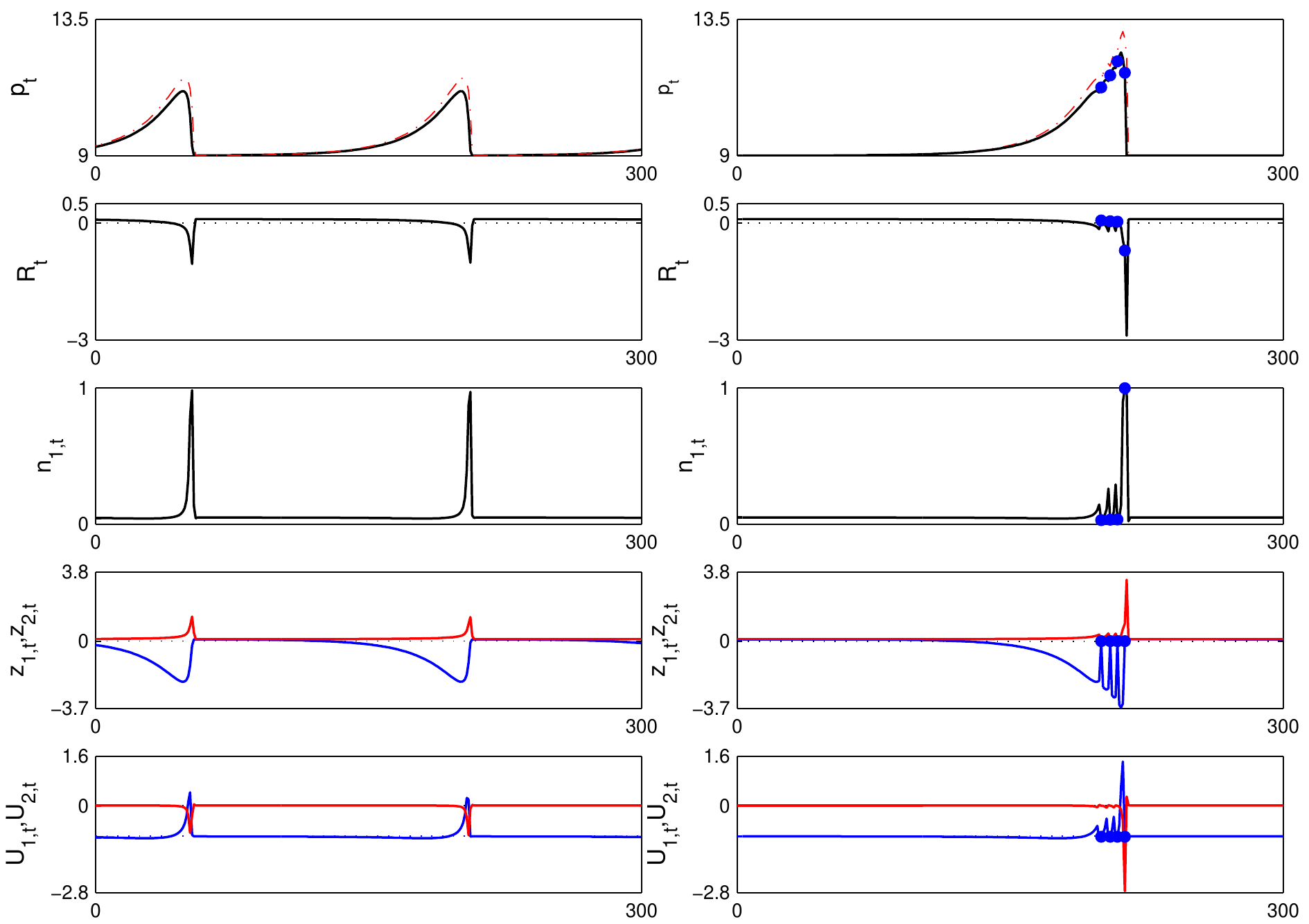}}
\caption{Positive (i.e. above fundamental) time series on the asymptotic regime of models (\ref{ARED_model_no_uptick_2},~\ref{eq:pr1},~\ref{eq:pr1g}) (left) and (\ref{ARED_model_uptick},~\ref{eq:pr1},~\ref{eq:pr1g}) (right) for $\beta=3$. Other parameter values as in Figure~(\ref{fig:BH_s})-case(B). The blue dots indicate that fundamentalists cannot have negative positions because of the uptick rule.}
\label{fig:BHtsp}
\end{figure}
The top panels report the price dynamics (black) and the chartist prediction (red, dashed), with blue and red dots marking the periods in which fundamentalists and chartists are respectively prevented to go short by the uptick rule (right column).
The remaining panels report, from top to bottom, the returns ($R_t=x_t-R\hspace{0.2mm}x_{t-1}+a\hspace{0.2mm}\sigma^2s$) and the traders' fractions ($n_{h,t}$), demands ($z_{h,t}$), and net profits ($U_{h,t}=R_t z_{h,t-1}-C_h$, $h=1,2$, blue for fundamentalists and red for chartists).

The price dynamics is characterized by recurrent peaks (financial bubbles), driven by chartists that expect a price rise and hold the shares of the risky asset and, at the same time, attenuated by fundamentalists which expect a devaluation and sell short (see the negative positions of fundamentalists in the unconstrained dynamics, left). In particular, when the price is closed to its fundamental value, the chartists' trading strategy is more profitable because their expectations are confirmed. Chartists dominate the market and the price is growing until the capital gain cannot compensate for the lower dividend yield. At this point, chartists start to suffer negative returns, while the short positions of the fundamentalists produce profits.
Eventually most of traders adopt the fundamentalist's trading strategy and the price falls down close to the fundamental value.
As soon the price starts to revert to the fundamental value, however, fundamentalists are prohibited from going short in the constrained dynamics because of the uptick rule (right column), and this triggers a further phase of rising prices. This happens several times, with the result of amplifying the price peak. At a certain point, when the price is very far away from its fundamental value, chartists suffer strong losses and the relative fraction of fundamentalists is almost one. The massive presence of fundamentalists pushes the price close to its fundamental value and the uptick rule cannot prevent this from happening. Indeed, if there are almost only fundamentalists, from the pricing equation we have that their demands must be equal to the supply of outside share, i.e. positive. When the price is close to its fundamental price chartists start to have a better performance and the story repeats. Despite this mispricing effect, the frequency of the price peaks is slowed down by the short selling restriction.

Similarly, Figure~(\ref{fig:BHtsn}) shows an example of time series on the negative attractor obtained for $\beta=4$ (here the unconstrained dynamics is quasi-periodic, see Figure~(\ref{fig:BHbif})).
\begin{figure}[t]
\centerline{\includegraphics[width=0.7\textwidth]{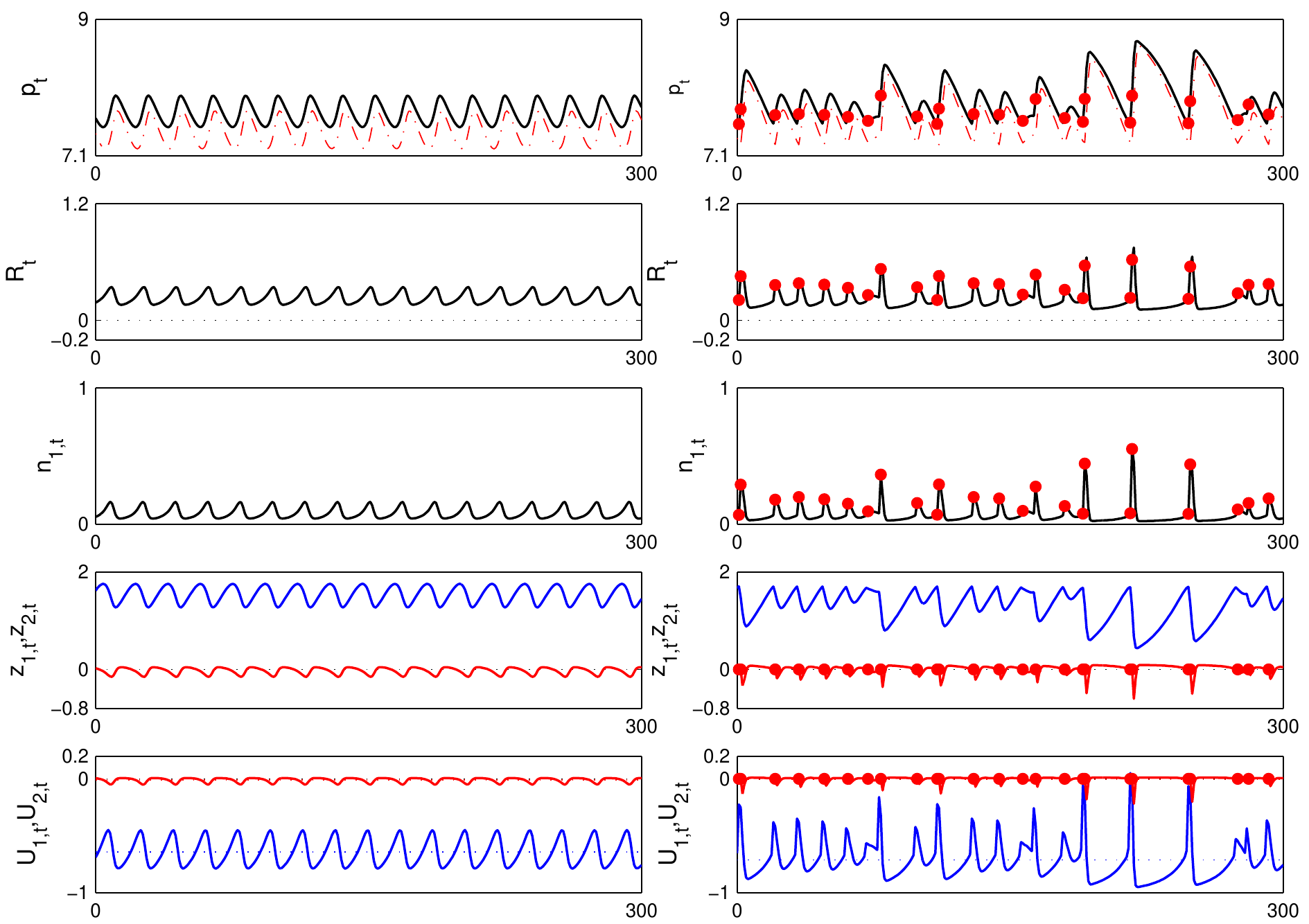}}
\caption{Negative (i.e., below fundamental) time series on the asymptotic regime of models (\ref{ARED_model_no_uptick_2},~\ref{eq:pr1},~\ref{eq:pr1g}) (left) and (\ref{ARED_model_uptick},~\ref{eq:pr1},~\ref{eq:pr1g}) (right) for $\beta=4$. Other parameter values as in Figure~(\ref{fig:BH_s})-case(B). The red dots indicate that chartists cannot have negative positions because of the uptick rule.}
\label{fig:BHtsn}
\end{figure}
In this case of negative price deviations, the chartists go short and have, on average, higher profits than fundamentalists. In the unconstrained dynamics chartists are predominant and drive prices down. This phenomenon is attenuated by the uptick rule, which limit the downward price movements and increases the performance of fundamentalists.
The frequencies of the negative peaks is slightly lowered by the short selling restriction, but they are more irregular due to the presence of chaotic dynamics as indicated by the positivity of the largest Lyapunov exponent, see Figure~(\ref{fig:BHbif}).

The different types of price dynamics, quasi-periodic for the unconstrained model and chaotic for the constrained one, are due to the different types of bifurcations through which the non-fundamental equilibrium $(\bar{x}^{(-)},\bar{m})$ losses stability, a Neimark-Sacker bifurcation in the first case and a border-collision bifurcation in the second. Indeed, as typical in piecewise linear model, at the border-collision bifurcation we have sudden transition from a stable fixed point to a fully developed {\it robust} (i.e., without periodic windows) chaotic attractor, see, e.g. \citep{Bernardo08}.

\subsection{Fundamentalists vs ROC traders}
\label{ssec:roc_n}
Figure~(\ref{fig:ROCbif}) reports the bifurcation analysis of models \eqref{ARED_model_no_uptick_2} (left) and \eqref{ARED_model_uptick} (right) with predictors \eqref{eq:pr1} and (\ref{eq:roc},c,d),
while Figure~(\ref{fig:ROCts}) shows the time series on the chaotic attractor obtained for $\beta=1.4$.
\begin{figure}[t]
\centerline{\includegraphics[width=0.7\textwidth]{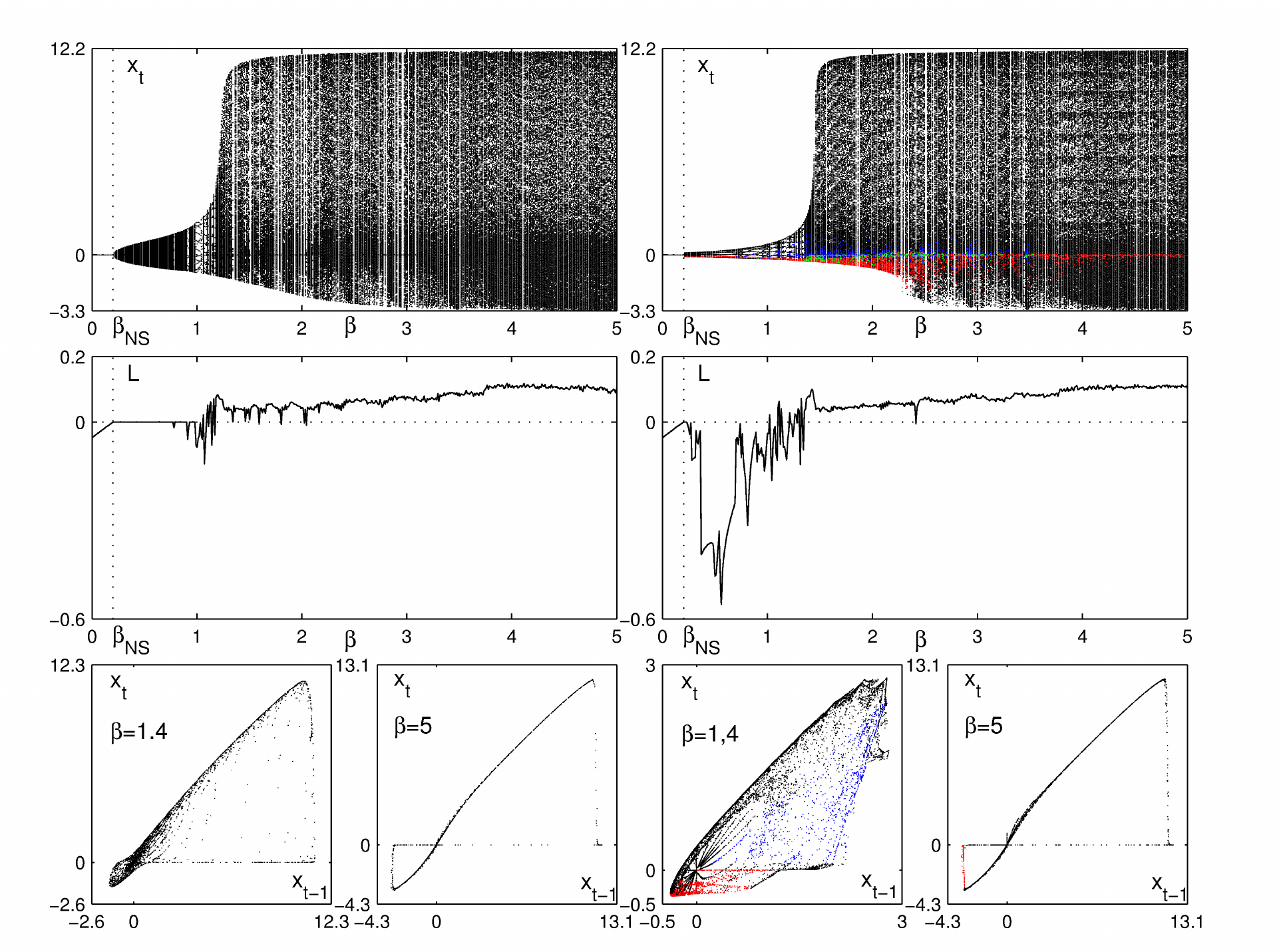}}
\caption{Bifurcation diagrams of models (\ref{ARED_model_no_uptick_2},~\ref{eq:pr1},~\ref{eq:roc},c,d) (left) and (\ref{ARED_model_uptick},~\ref{eq:pr1},~\ref{eq:roc},c,d) (right) and exemplary state portraits projections in the plane $(x_{t},x_{t-1})$.
Parameter values: $\alpha=10$, $L=2$, other values as in Figure~(\ref{fig:BH_s})-case(B). The blue dots indicate that fundamentalists cannot have negative positions because of the uptick rule. The red dots indicate that chartists cannot have negative positions because of the uptick rule.}
\label{fig:ROCbif}
\end{figure}

\begin{figure}[t]
\centerline{\includegraphics[width=0.7\textwidth]{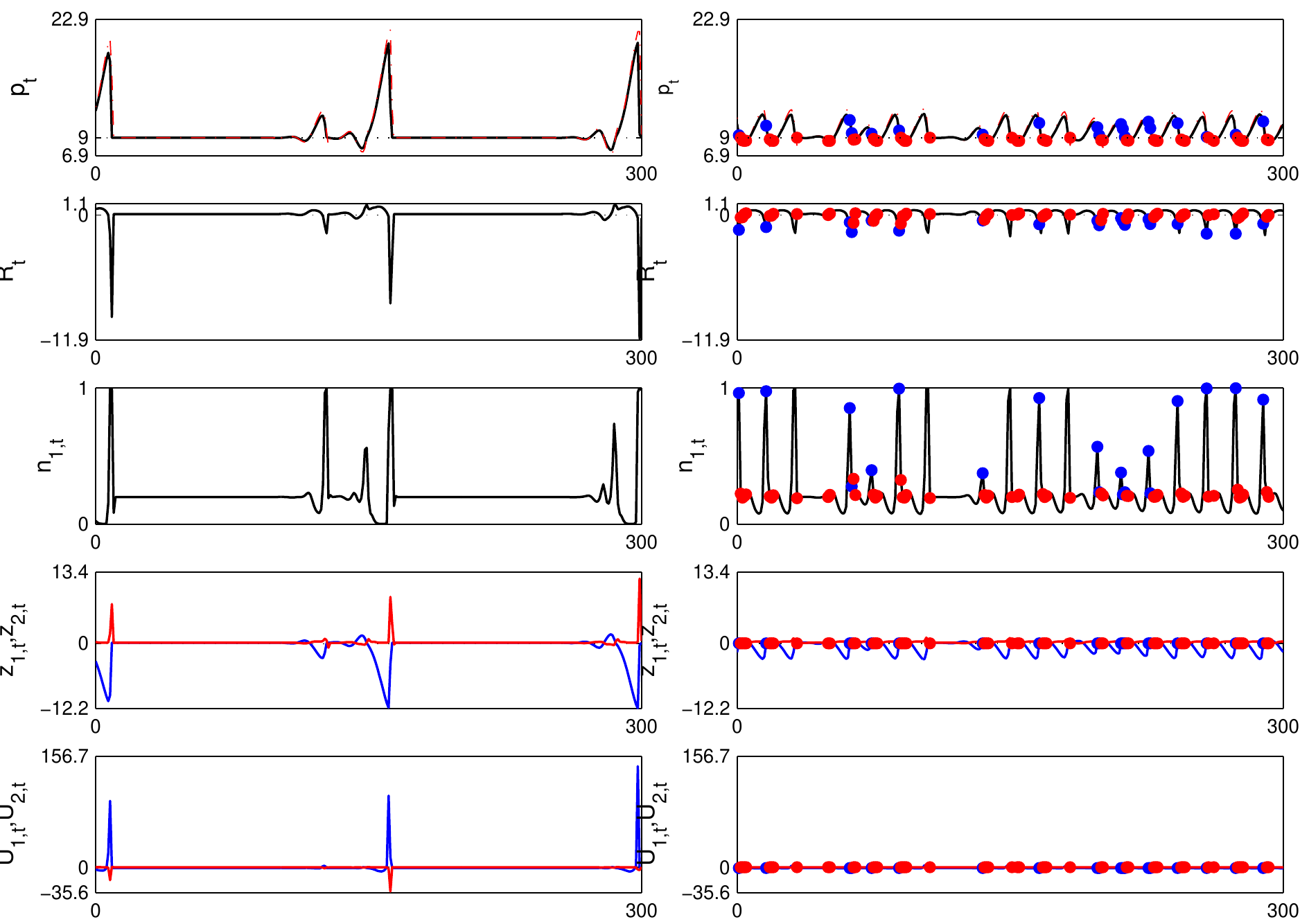}}
\caption{Time series on the asymptotic regime of models (\ref{ARED_model_no_uptick_2},~\ref{eq:pr1},~\ref{eq:roc},c,d) (left) and (\ref{ARED_model_uptick},~\ref{eq:pr1},~\ref{eq:roc},c,d) (right) for $\beta=1.4$. Other parameter values as in Figure~(\ref{fig:ROCbif}). The blue dots indicate that fundamentalists cannot have negative positions because of the uptick rule. The red dots indicate that chartists cannot have negative positions because of the uptick rule.}
\label{fig:ROCts}
\end{figure}

Here the fundamental equilibrium is globally stable up to the Neimark-Sacker bifurcation at $\beta\leq\beta_{\mathrm{NS}}$.
Interestingly, the price fluctuations showed by the non-stationary attractors originated for larger $\beta$ (quasi-periodic with periodic windows and later chaotic) have a remarkably smaller amplitude in the constrained dynamics, than in the unrestricted case.
In this sense, the uptick-rule shows a rather robust stabilizing effect, at least as long $\beta$ is not too large, i.e., traders are not fast enough in changing their beliefs to react to past performances.

This is also evident in the time series of Figure~(\ref{fig:ROCts}), where the short selling restriction also intensifies the frequency of the price peaks.

\begin{figure}[t]
\centerline{\includegraphics[width=0.7\textwidth]{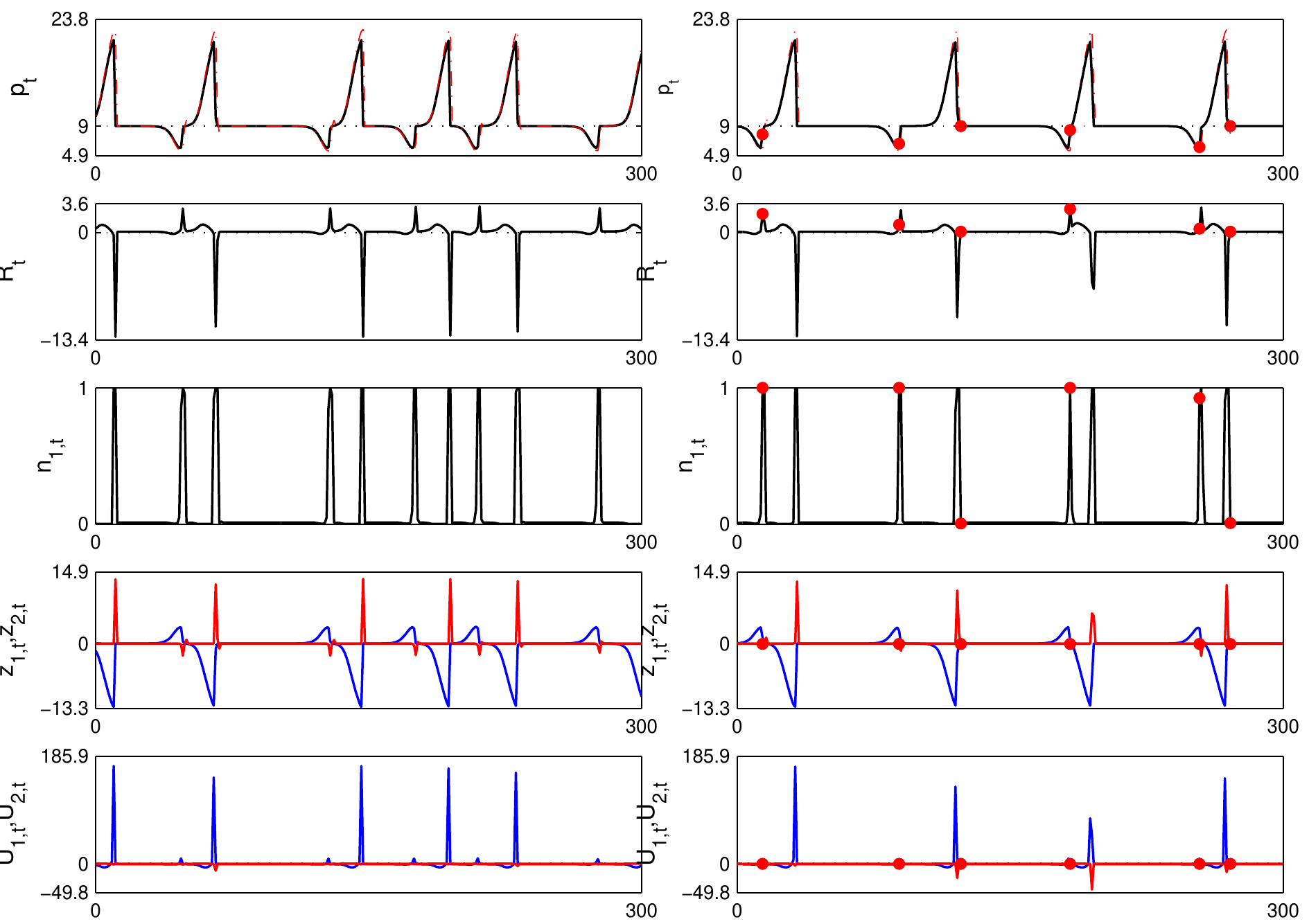}}
\caption{Time series on the asymptotic regime of models (\ref{ARED_model_no_uptick_2},~\ref{eq:pr1},~\ref{eq:roc},c,d) (left) and (\ref{ARED_model_uptick},~\ref{eq:pr1},~\ref{eq:roc},c,d) (right) for $\beta=4.5$. Other parameter values as in Figure~(\ref{fig:ROCbif}). The red dots indicate that chartists cannot have negative positions because of the uptick rule.}
\label{fig:ROCts2}
\end{figure}

\subsection{Economic discussion of the numerical results}
\label{sec:ed}
On the basis of the three declared objectives of the uptick rule, this subsection provides a discussion of the effects of this short selling regulation
in the light of the analytical and by numerical analysis reported in the previous Sections.


Let us start to discuss the scenario described by models \eqref{ARED_model_no_uptick_2} and \eqref{ARED_model_uptick} with predictors \eqref{eq:pr1v} and \eqref{eq:pr1g}
(fundamentalists vs chartists). We first consider the case of negative price deviations (the negative attractor in Figure~(\ref{fig:BHbif}) and the time series in Figure~(\ref{fig:BHtsn})), where we see that chartists take short positions
when prices fall, whereas fundamentalists take long positions
believing that the stock price will rise to reach the fundamental value
(see the unconstrained dynamics in the left column of Figure~(\ref{fig:BHtsn})).
It is important to point out that whenever an asset is undervalued (price below its fundamental value) it should be better to hold it rather than to sell it because the return is always positive: dividend yield outweighs the capital gain effect. Indeed, by going short chartists obtain negative excess returns, see the dynamic of the net profits $U_{2,t}$ in the last row of Figure~(\ref{fig:BHtsn}). However chartists do better than fundamentalists here because the latter are charged a cost, $\left(C\right)$, which is larger than the profit they obtain by holding the assets, compare the net profits of the two trading strategies ($U_{1,t}$ and $U_{2,t}$) again in the last row of Figure~(\ref{fig:BHtsn}). It follows that
chartists are predominant in the market, as indicated by the low value of the fraction $n_1$,
and their trading strategy causes downward movements of the stocks' price.
In the constrained dynamics (right column), the uptick rule limits the possibility to go
short for chartists. This helps to revert the stock prices toward the fundamental value. It follows a better performance for the fundamentalists and their presence in the market increases. As a result, the short selling restriction reduces the negative peaks reached by stock prices.
From this example it is clear the effectiveness of the regulation to meet the last two goals established by the SEC.
Moreover, from the analysis of the bifurcation diagram, it is interesting to note that the negative non-fundamental equilibrium $(\bar{x}^{(-)},\bar{m})$
looses stability in the constrained dynamics at a lower traders' adaptability, i.e., at a lower value of the parameter $\beta$, respect to the unconstrained dynamics.
This is due to the border collision bifurcation at $\beta=\beta_{\mathrm{BC}}^{(-)}$, that occurs before the Neimark-Sacker bifurcation at $\beta_{\mathrm{NS}}^{(-)}$ responsible of the instability in the unconstrained dynamics.
This might be mathematically interpreted as a destabilizing effect of the uptick rule. However, the chaotic fluctuations established after the bifurcation move the prices, on average, closer to the fundamental value with respect to the equilibrium deviation $\bar{x}^{(-)}$.

Considering the same model and predictors, by the time series of Figure~(\ref{fig:BHtsp}) it is
possible to notice that the dynamics of positive price deviations are
characterized (for $\beta>\beta_{\mathrm{NS}}^{(+)}$) by frequently financial bubbles.
These bubbles are enforced by chartists which overvalued the stock prices
in the upward trend and are attenuated by fundamentalists. In fact,
in these phases of rising prices fundamentalists go short driven by the
belief that the price will revert to its fundamental value in the next
period. This increases the supply of shares helping to curb rising prices. At a certain point of the upward trend, the fundamentalists' strategy
take over and prices are driven to the fundamental value of the stock.
During these upward trends of the market the uptick rule prevents the fundamentalists to take short positions, see blue dots in the fourth row,
right column of Figure~(\ref{fig:BHtsp}).
This phenomenon is also emphasized by empirical
analysis (see, e.g., \citep{AlexanderPeterson1999} and  \citep{BoehmerJonesZhang2008}), and represents a flaw of
the regulation.
The result is
an increase of the amplitude of the financial bubbles. As a positive effect
of the regulation, it is possible to notice a reduction in the frequency of
occurrences of these bullish divergences.

The main point is that every time fundamentalists go short they force prices
to converge to the fundamental value. If it were possible to
discriminate between the beliefs of the traders, fundamentalists should not be
forbidden to go short in this situation. However, this is a very difficult task and it is not even so easy to correctly determine the
fundamental values of a risky asset in the real market. A possible solution could be
to allow short sales after a long period of rising prices.
This should avoid pushing up prices because fundamentalists (and
contrarians, but these traders are not taken into consideration in this
work) are driven out
of the market. Another possible solution could be to restrict short selling
only in cases of sharp and sudden falls in stock prices. This should produce
two effects, to let agents believing in the fundamental price go short when the
price increases, reducing positive oscillations of price deviations from
the fundamental value, and to reduce sharp drops in prices observed
when the stock market bubble breaks. An important step in this direction has
already been done, the SEC adopted Rule 201 which was implemented on
February 28, 2011%
\footnote{%
See, Securities Exchange Act Release No. 61595 (Feb. 26, 2010).}.
This new short selling regulation prohibits short selling operations if the value of
the stock decreases by more than 10\% in two consecutive trading sections.

Summarizing the analysis of the positive and negative price deviations from
the fundamental value for the model with predictors \eqref{eq:pr1} and \eqref{eq:pr1g} for the cases
of unrestricted and restricted short sales, we can conclude that the second goal 
of the regulation is ensured, but some distorting effects produced by this rule are
observed, such as overvaluation of shares. Moreover, the short selling restriction can trigger distorting
mechanisms that support dynamics of overpricing, otherwise not feasible in
the long run. This is indicted by the presence of multiple attractors in the region of positive price deviations
in the bifurcation diagram in the right column of Figure~(\ref{fig:BHbif}).

With the intent to provide a more detailed and comprehensive description of the effects of the
uptick rule, the same asset pricing model has been analyzed with a
different non-fundamental predictor (the Smoothed-ROC predictor, see Subsection~\ref{ssec:pre}).
Differently from other trend-following indicators, the Smoothed-ROC
predictor is particularly useful to detect fast and short-term upward and
downward movements of the stock price and gives different trading signals
to agents than predictor \eqref{eq:pr1g}, such as gain and lose of speed in the trend
(see \citep{Elder1993}). In this case, we obtain an interesting and surprising
result, i.e., the uptick rule
helps to reduce the amplitude of price fluctuations, and causes an increase
of the frequency of oscillations
above and below the fundamental value, this is made clear by comparing the
time series of prices in the first row of Figure~(\ref{fig:ROCts}). The explanation of this lies on the higher degree of rationality (compared to the one assumed for chartists) of the non-fundamentalist agents that the use of the non-linear predictor (S-ROC predictor) implies. These agents are uncomfortable with extreme assessments of the value of the shares and when the shares are forced to be overvalued due to the short selling constraints, non-fundamental agents react and become more confidential in the fundamental price. This changes their trading strategies and, as a results, the amplitude of the price fluctuations decreases instead of increasing as might be expected. The choice of the value of $\alpha$ plays an important rule in this.
We can conclude that by using this
couple of predictors it is possible to observe, at least when evolutionary pressure, $\beta$, is not excessively large,
all the main goals of the uptick rule:
short selling restriction does not produce mispricing, prevents chartists from going short during downward price movements
to avoid reaching negative price variation peaks and prevents fundamentalists from 
going short only in a downward price movement for positive price
deviations, reducing the speed of convergence to the fundamental value and
preventing sharps drops in prices. As a further observation, it is worth noticing that the uptick rule prevents extremely negative excess returns and at the same time makes the strategy of fundamentalists on the average more profitable, changing the beliefs of traders, compare $n_{1,t}$ and $n_{2,t}$ in Figure~(\ref{fig:ROCts}). Moreover, looking at the bifurcation diagram of Figure~(\ref{fig:ROCbif}) right-column, it is clear that for relatively high values of the intensity of choice, $\beta$, the uptick rule does not have any effect on the amplitude of price fluctuations. As revealed by the analysis of the time series in Figure~(\ref{fig:ROCts2}), this is due to the fact that for high levels of $\beta$ the market is always dominated by one trading strategy. It follows that the short selling constraint can apply either to a trading strategy adopted almost by any trader or to a trading strategy adopted by almost all the traders. In the first case the regulation does not have any effect on the price, in the second case it reverts the price toward the fundamental value with the results of reducing the frequency of market-bubbles but without reducing the amplitude of them. Compare the two panels in the first row of Figure~(\ref{fig:ROCts2}). For $\beta$ large enough, at each trading section there is only one type of trader that dominates the market and its demand of shares, being equal to the supply of outside share per trader, must be positive, then the uptick rule does not affect the dynamics of prices. 


The conducted analysis points out that the uptick rule meets part of its
goals, but which of them often depends on the condition of the markets.
Moreover, the regulation can produce several side-effects which may be
different according once again to the market's conditions and investor's beliefs which may
strongly influence the effectiveness of the regulation itself. Due to these
findings, studying the impact of the uptick rule on financial markets does not seem to be
an easy
task. Nonetheless, it is possible to isolate some remarkable effects
regardless of market conditions and traders' beliefs. First, the uptick rule
ensures a
reduction of the downward market movements when the shares are undervalued, i.e. when the shares are priced below their fundamental value.
Second, the intensity of choice $\beta$ to switch predictors affects the
effectiveness of the short selling regulation. By using the bifurcation analysis, it
possible to observe that the effectiveness of the regulation fades away
increasing the value of $\beta$
(increasing $\beta$, agents tend to overreact to the market's information,
changing their beliefs quickly to react to past performances).
In other words, the switching destabilizing effect prevails over the regulation's
effects, i.e. when agents overreact to the differences in performance related to different beliefs, the regulation
does not affect the dynamics of stock prices. This is consistent with
many interesting empirical results testifying that there is no statistical
effect of the uptick rule on price fluctuations in
turbulent financial markets (see, e.g., \citep{DietherLeeWerner2009}).

As a final remark, it is important to clarify that the simple asset pricing model
used in this paper can reproduce only some possible ``stylized effects'' which are a direct
consequence of the regulation. However, in the real financial markets many more
different trading strategies and emotional actions are present, which can modify the effectiveness of the uptick rule.


\section{Conclusions and future directions}
\label{sec:cd}
In this paper, we have investigated the effect of the "uptick rule" on an asset price dynamics by means of an asset pricing model with heterogeneous, adaptive beliefs. The analysis has suggested the effectiveness of the regulation in reducing the downward price movements of undervalued shares avoiding speculative behavior whenever the market is characterized by not too many aggressive traders, i.e. when the agents' propensity of changing trading strategy is relatively low. On the contrary, when the agents have a high propensity of changing trading strategy according to the past trading performances, which is a sign of turbulent markets according to the model, the effects of the regulation tend to fade. As a side effect of the regulation, an amplification of the market bubbles in the case of overvalued shares is possible.

This work represents only a starting point. There are still several aspects
that can and deserve to be analyzed. First of all, it is interesting to
analyze the effect of the uptick regulation using the same asset pricing
model with a greater number of investor types, for example contrarians,
chartists and fundamentalists. In fact, there is empirical evidence
in the literature about the switch in trading style by short-sellers.
\citep{DietherLeeWerner2009} found that under the uptick rule most of the
short positions are opened by contrarians, on the contrary, when the uptick
rule is not imposed are chartists, the ones who prefer to go short
(see, e.g., \citep{BoehmerJonesZhang2008}). There is a hypothetical
explanation for this. Chartists take short positions in declining
price trends and the uptick rule makes this operation more difficult, on the
contrary the restriction does not effect the contrarians' short strategy.
They usually go short in upward price trend betting on a change in price
movement with the effect of stabilizing the market. Investigating the validity
of this hypothesis provides a better understanding of the issue.

The regulation should also be evaluated in the contest of the multi-assets
market to discover how the short selling restriction for one stock influences
the price of the others. It is reasonable to expect that traders will switch
to trade stocks that are not effected by the restriction in that specific
moment and this can produce effects on prices which are not easy to predict
without a deep analysis.

Another aspect that deserves to be investigated is the effect of the
regulation when there is a fraction of investors which does not change the
trading strategy (or belief) as in \citep{DieciForoniGardiniHe2006}. As highlighted in this
paper, the uptick rule loses its effectiveness due to a high propensity to switch trading strategy by agents. It follows that, if there are
constraints on the possibility to change trading strategy, we expect an increase of the effectiveness of the regulation and a reduction of unwelcome effects on price dynamics.  Last but not least, the piecewise continuous model here
proposed can be easily adapted according to the new short selling regulation
imposed by the SEC, i.e. Rule 201. Comparing the two cases can help
to understand the pros and cons of the new regulation.

\renewcommand{\theequation}{A\arabic{equation}} 
\setcounter{equation}{0}  
\section{Appendix}\label{Ch4Appendix}

\begin{proof}[Proof of Lemma \ref{lm:fund}]
The existence of the fundamental equilibrium immediately follows by substituting $f_h(\mathbf{0})=0$ into Eqs.~\eqref{ARED_model_no_uptick} and~\eqref{ARED_model_no_uptick_2}.
For the definition of the associated eigenvalues, let us substitute Eq.~\eqref{ARED_model_no_uptick_n} written for $n_{h,t}$ into Eq.~\eqref{ARED_model_no_uptick_x}
(and Eq.~\eqref{ARED_model_no_uptick_2m} written for $m_{t}$ into Eq.~\eqref{ARED_model_no_uptick_2x}).
Then, we can consider $(x_{t-1},\ldots,x_{t-(L+2)})$ as the state variables for both models, and the Jacobian of the system at the fundamental equilibrium is equal to
\begin{linenomath}
$$
\left[
\begin{array}{c|c}
\begin{array}{ccccc}
\gamma_1 & \hspace{4mm} \gamma_2 \hspace{4mm} & \cdots & & \gamma_L\\
1 & 0 & \cdots & & 0\\
0 & 1 & \cdots & & 0\\
\vdots & & \hspace{-1mm}\ddots & & 0\\
0 & 0 & \cdots & \hspace{5mm} 1 \hspace{3mm} & 0
\end{array} &
\begin{array}{ccc}
\hspace{1mm} 0 \hspace{1mm} & & \hspace{1mm} 0 \hspace{1mm}\\
0 & & 0\\
0 & & 0\\
\vdots & & \vdots\\
0 & & 0
\end{array}\\
\hline
\begin{array}{ccccc}
\hspace{-0.2mm} 0 \hspace{1.9mm} & \hspace{4mm} 0 \hspace{5mm} & \cdots & \hspace{5mm} 0 \hspace{2mm} & \hspace{2mm} 1\\
\hspace{-0.2mm} 0 \hspace{1.9mm} & \hspace{4mm} 0 \hspace{5mm} & \cdots & \hspace{5mm} 0 \hspace{2mm} & \hspace{2mm} 0
\end{array} &
\begin{array}{ccc}
\hspace{1mm} 0 \hspace{1mm} & & \hspace{1mm} 0 \hspace{1mm}\\
1 & & 0
\end{array}
\end{array}\right],
$$
\end{linenomath}
which proves the result.
\end{proof}\\

\begin{proof}[Proof of Lemma \ref{lm:neq}]
At the equilibrium price deviation $\bar{x}$, we obviously have $\mathbf{x}_{t}=\bar{x}\mathbf{1}$ and, from Eq.~\eqref{ARED_model_no_uptick_x}, $\bar{x}$ must satisfy
\begin{linenomath}
$$
R = \sum_{h=1}^{H}n_{h,t}\hspace{0.2mm}f_h(\bar{x}\mathbf{1})/\bar{x}.
$$
\end{linenomath}
Being $n_{h,t}\in(0,1)$ for all $h=1,\ldots,H$, this is possible only if $f_h(\bar{x}\mathbf{1})/\bar{x}<R$ for some $h$ and $f_k(\bar{x}\mathbf{1})/\bar{x}>R$ for some $k\neq h$.
\end{proof}\\


\begin{proof}[Proof of Lemma \ref{lm:cha}]
\begin{enumerate}
\item
The uniqueness follows from Lemma \ref{lm:neq}, while the global stability from \eqref{ARED_model_no_uptick_2x}, which can be rewritten as
\begin{linenomath}
$$
x_{t} = \Frac{1}{R}\left(n_{1,t}\hspace{0.2mm}v + n_{2,t}\hspace{0.2mm}g\right)x_{t-1},
$$
\end{linenomath}
and contracts the deviation $x_{t}$ as $t$ goes to infinity.
\item
From Lemma \ref{lm:fund}, the non-zero eigenvalue associated to the fundamental equilibrium is
\begin{linenomath}
$$
\lambda^{(0)}=\gamma_1=\Frac{1}{2R}\left((1+\bar{m}^{(0)})\hspace{0.2mm}v + (1-\bar{m}^{(0)})\hspace{0.2mm}g\right)>0,\quad
\bar{m}^{(0)}=\tanh\left(-\beta\hspace{0.2mm}C/2\right).
$$
\end{linenomath}
Thus the fundamental equilibrium can loose stability only when $\lambda^{(0)}=1$ at a transcritical (or pitchfork) bifurcation
(being fixed point of model (\ref{ARED_model_no_uptick_2},~\ref{eq:pr1v},~\ref{eq:pr1g}) for any admissible parameter setting, it cannot disappear through a saddle-node bifurcation).
Solving $\lambda^{(0)}=1$ for $\beta$ gives $\beta_{\mathrm{TR}}$.

Evaluating Eqs.~\eqref{ARED_model_no_uptick} at the generic equilibrium $(\bar{x},\bar{m})$, solving Eq.~\eqref{ARED_model_no_uptick_2x} for $\bar{m}$, and equating the result to Eq.~\eqref{ARED_model_no_uptick_2m}, we get
\begin{linenomath}
$$
\bar{m}= 1 - 2\hspace{0.2mm}\Frac{R-v}{g-v} =
\tanh\left(-\Frac{\beta}{2}\left((-(R-1)\hspace{0.2mm}\bar{x}+a\hspace{0.2mm}\sigma^{2}s)\hspace{0.2mm}
\Frac{g-v}{a\hspace{0.2mm}\sigma^{2}}\hspace{0.2mm}\bar{x}+C\right)\right),
$$
\end{linenomath}
which solved for $\bar{x}$ gives $\bar{x}^{(\pm)}$.

Solving $\bar{x}^{(+)}=\bar{x}^{(-)}$ for $\beta$ gives $\beta_{\mathrm{LP}}$ and the equilibrium deviations $\bar{x}^{(\pm)}$ are defined only for $\beta\hspace{-0.2mm}>\hspace{-0.2mm}\beta_{\mathrm{LP}}$.
There are therefore no other equilibria and this concludes the proof of points (a), (LP), and (TR).
Note that $\beta_{\mathrm{LP}}=\beta_{\mathrm{TR}}$ when $s=0$ (the transcritical and saddle-node bifurcations coincide at a pitchfork bifurcation).

Substituting Eq.~\eqref{ARED_model_no_uptick_2m} written for $m_{t}$ into Eq.~\eqref{ARED_model_no_uptick_2x} and using $(x_{t-1},x_{t-2},x_{t-3})$ as state variables, the Jacobian of the systems  at equilibria $(\bar{x}^{(\pm)},\bar{m})$ is given by
\begin{linenomath}
$$
\left[\begin{array}{ccc}
\gamma_1^{(\pm)} & \gamma_2^{(\pm)} & \gamma_3^{(\pm)}\\
1 & 0 & 0\\
0 & 1 & 0
\end{array}\right],
$$
\end{linenomath}
with
\begin{linenomath}
$$
\gamma_1^{(\pm)} = 1 + \Frac{\gamma}{R}\hspace{0.2mm}\bar{x}^{(\pm)},\quad
\gamma_2^{(\pm)} = -\gamma\hspace{0.2mm}\bar{x}^{(\pm)},\quad
\gamma_3^{(\pm)} = \Frac{\gamma}{R}\left(-(R-1)\hspace{0.2mm}\bar{x}^{(\pm)}+a\hspace{0.2mm}\sigma^{2}s\right),
$$
\end{linenomath}
and
\begin{linenomath}
$$
\gamma = \beta\hspace{0.2mm}\bar{x}^{(\pm)}
\Frac{(g-v)^2}{4\hspace{0.2mm}a\hspace{0.2mm}\sigma^{2}}\hspace{0.5mm}
\mathrm{sech}^{\hspace{-0.1mm}2}\hspace{-0.8mm}\left(\hspace{-0.5mm}-\Frac{\beta}{2}\left((-(R-1)\hspace{0.2mm}\bar{x}^{(\pm)}+a\hspace{0.2mm}\sigma^{2}s)\hspace{0.2mm}
\Frac{g-v}{a\hspace{0.2mm}\sigma^{2}}\hspace{0.2mm}\bar{x}^{(\pm)}+C\right)\hspace{-0.5mm}\right),
$$
\end{linenomath}
so the three associated eigenvalues $\lambda_1$, $\lambda_2$, $\lambda_3$ are the roots of the characteristic equation
\begin{linenomath}
$$
\lambda^3-\gamma_1^{(\pm)}\lambda^2-\gamma_2^{(\pm)}\lambda-\gamma_3^{(\pm)}=0.
$$
\end{linenomath}

In particular, imposing $\lambda=1$ and solving the characteristic equation for $\bar{x}^{(\pm)}$, gives only zero and $x_{\mathrm{LP}}$ as solutions, so no other transcritical, saddle-node (or pitchfork) bifurcation is possible.
In contrast, imposing $\lambda_1=-1$ and taking into account that $\lambda_1+\lambda_2+\lambda_3=\gamma_1$, we get $\lambda_2+\lambda_3=2+\gamma\hspace{0.2mm}\bar{x}^{(\pm)}/R>2$ (note that $\gamma\hspace{0.2mm}\bar{x}^{(+)}>0$ and that $\gamma\hspace{0.2mm}\bar{x}^{(-)}\ge 0$ only vanishes at the transcritical bifurcation), so that equilibrium $(\bar{x}^{(\pm)},\bar{m})$ would be unstable at a period-doubling (flip) bifurcation.
However, both equilibria $(\bar{x}^{(\pm)},\bar{m})$ loose stability by increasing $\beta$, because the coefficient $\gamma_1$ linearly diverges with $\beta$
(the limit as $\beta\to\infty$ of the $\mathrm{sech}$ argument is finite and equal to $-\log((R\hspace{-0.2mm}-\hspace{-0.2mm}v)/(g\hspace{-0.2mm}-\hspace{-0.3mm}R))/2<0$), so the same does the sum of the eigenvalues.
Stability is therefore lost through a Neimark-Sacker bifurcation and this concludes the proof of points (b), (c), and (NS$^{(\pm)}$).
\item
For $g>2R-v$, the non-zero eigenvalue $\lambda^{(0)}$ associated to the fundamental equilibrium is larger than one for any $\beta>0$.
We also have $\beta_{\mathrm{LP}}<0$, $\beta_{\mathrm{TR}}<0$, and equilibria $(\bar{x}^{(\pm)},\bar{m})$ are defined for any $\beta>0$ with $\bar{x}^{(+)}>0$ and $\bar{x}^{(-)}<0$.
Similarly to point 2, they loose stability through a Neimark-Sacker bifurcation.
\item
In the limiting case $\beta\to\infty$, from Eq.~\eqref{ARED_model_no_uptick_2m} we have
\begin{linenomath}
$$
m_{t+1}=\left\{\begin{array}{rl}
1  & \text{if}\quad
(R\hspace{0.2mm}x_{t-1}-x_{t}-a\hspace{0.2mm}\sigma^{2}s)
\Frac{g-v}{a\hspace{0.2mm}\sigma^{2}}x_{t-1}>C,
\\
-1 & \text{if}\quad
(R\hspace{0.2mm}x_{t-1}-x_{t}-a\hspace{0.2mm}\sigma^{2}s)
\Frac{g-v}{a\hspace{0.2mm}\sigma^{2}}x_{t-1}<C.
\end{array}\right.
$$
\end{linenomath}
Starting at $x_0=\pm\eps$ with sufficiently small $\eps>0$ and $m_1=-1$, we therefore have $x_t=\pm\eps(g/R)^t$ (see Eq.~\eqref{ARED_model_no_uptick_2x}) as long as $m_{t+1}$ stays at $-1$.
Thus, $x_t$ diverges if
\begin{linenomath}
$$
\eps^2(g/R)^{2t-1}\left(\Frac{R^2}{g}-1 \mp a\hspace{0.2mm}\sigma^{2}s\hspace{0.2mm}\eps^{-1}(g/R)^{-t}\right)
\Frac{g-v}{a\hspace{0.2mm}\sigma^{2}}>C,
$$
\end{linenomath}
is never satisfied for increasing $t$, i.e., when $g>R^2$.
\end{enumerate}
\end{proof}\\

\begin{proof}[Proof of Lemma \ref{lm:cha_uptick}]
\begin{enumerate}
\item
At the fundamental equilibrium, the traders' demands are positive ($\bar{z}^{(0)}_1=\bar{z}^{(0)}_2=s$), so $(0,\bar{m}^{(0)})$ is an admissible fixed point of the unconstrained dynamics for any admissible parameter setting. Its local stability is therefore ruled by Lemma~\ref{lm:cha}.

The global stability for $1<g<R$ is a consequence of the following arguments.
First, the price deviation is contracted from period $(t-1)$ to period $t$ (i.e., $|x_t|<|x_{t-1}|$) whenever the unconstrained dynamics is applied (see Lemma~\ref{lm:cha}, point 1).
Second, the unconstrained deviation $x_t^{(0)}$ following $x_{t-1}$ is smaller than any of the constrained deviations $x_t^{(1)}$ and $x_t^{(2)}$ given by Eq.~\eqref{ARED_model_uptick_x} in region $Z_1$ and $Z_2$, respectively.
This is graphically clear from Figure~(\ref{fig:mc}), and is analytically shown by noting that
\begin{linenomath}
$$
x_t^{(h)}\hspace{-0.2mm}-x_t^{(0)} = -\Frac{a\hspace{0.2mm}\sigma^{2}}{R}\hspace{0.2mm}\Frac{n_h}{n_k}\hspace{0.2mm}z_{h,t}^{(0)}>0,\quad h=1,2,\; k\neq h
$$
\end{linenomath}
(see \eqref{eq:regz}, and recall that $z_{h,t}^{(0)}<0$ in region $Z_h$).
Third, from Eq.~\eqref{ARED_model_uptick_x} we get that $0<x_t^{(0)}<x_t^{(h)}<(g/R)x_{t-1}$ when $x_{t-1}>0$ and $x_t^{(0)}<x_t^{(h)}<(v/R)x_{t-1}<0$ when $x_{t-1}<0$.
Thus, the constrained dynamics in regions $Z_1$ and $Z_2$ also contracts the price deviation from period $(t-1)$ to period $t$.

Being $\bar{x}^{(+)}>0$, equilibrium $(\bar{x}^{(+)},\bar{m})$ can only collide with border $\partial Z_1$ (see \eqref{eq:bor}), at which $\bar{x}^{(+)}=\bar{x}_{\mathrm{BC}}^{(+)}$ (the expression for $\bar{x}_{\mathrm{BC}}^{(+)}$ can be easily obtained by solving the first equation in \eqref{eq:regz} with $z_{1,t}^{(0)}=0$ for $x_{t-1}$).
It is obviously admissible iff $\bar{x}^{(+)}\le\bar{x}_{\mathrm{BC}}^{(+)}$.

Depending on the parameter setting, equilibrium $(\bar{x}^{(-)},\bar{m})$ can be either positive or negative, and can therefore collide with both borders $\partial Z_1$ and $\partial Z_2$.
At the border $\partial Z_2$, $\bar{x}^{(-)}=\bar{x}_{\mathrm{BC}}^{(-)}$ (the expression for $\bar{x}_{\mathrm{BC}}^{(-)}$ is obtained by solving the second equation in \eqref{eq:regz} with $z_{2,t}^{(0)}=0$ for $x_{t-1}$), so that $(\bar{x}^{(-)},\bar{m})$ is admissible iff  $\bar{x}_{\mathrm{BC}}^{(-)}\le\bar{x}^{(-)}\le\bar{x}_{\mathrm{BC}}^{(+)}$.
\item
If $R+v>2$ (case (a)), $\bar{x}_{\mathrm{LP}}$ from Lemma~\ref{lm:cha} is smaller than $\bar{x}_{\mathrm{BC}}^{(+)}$, so that equilibria $(\bar{x}^{(\pm)},\bar{m})$ are admissible at the saddle-node bifurcation.
The price deviation $\bar{x}^{(+)}$ increases as $\beta$ increases and reaches $\bar{x}_{\mathrm{BC}}^{(+)}$ at $\beta=\beta_{\mathrm{BC}}^{(+)}$ (collision with border $\partial Z_1$).
If $R+v<2$ (case (b)), equilibria $(\bar{x}^{(\pm)},\bar{m})$ are virtual at the saddle-node bifurcation.
The price deviation $\bar{x}^{(-)}$ decreases as $\beta$ increases and reaches $\bar{x}_{\mathrm{BC}}^{(+)}$ at $\beta=\beta_{\mathrm{BC}}^{(+)}$.
If $R+v=2$ (case (c)), then $\beta_{\mathrm{BC}}^{(+)}=\beta_{\mathrm{LP}}$.

Equilibrium $(\bar{x}^{(-)},\bar{m})$ collides with the border $\partial Z_2$ only if the limit of $\bar{x}^{(-)}$ as $\beta\to\infty$ is below $\bar{x}_{\mathrm{BC}}^{(-)}$.
This yields the condition on $s$ and the border collision at $\beta=\beta_{\mathrm{BC}}^{(-)}$ in point (d).
\item
For $g>2R-v$, $\bar{x}^{(+)}$ increases as $\beta$ increases, whereas $\bar{x}^{(-)}$ decreases, and their limiting value for $\beta\to\infty$ are as in point 2.
Equilibrium $(\bar{x}^{(+)},\bar{m})$ is always virtual, because its limiting value is above $\bar{x}_{\mathrm{BC}}^{(+)}$ for any admissible parameter setting.
Equilibrium $(\bar{x}^{(-)},\bar{m})$ becomes admissible at $\beta=\beta_{\mathrm{BC}}^{(-)}$ only if its limiting value is above $\bar{x}_{\mathrm{BC}}^{(-)}$, which gives the conditions at points (a) and (b).
\item
In the limiting case $\beta\to\infty$, $m_t$ switches between $\pm 1$, so that only one type of trader is present and the three options in Eq.~\eqref{ARED_model_uptick_x} give the same price deviation $x_t$. The result therefore follows from Lemma~\ref{lm:cha}.
\end{enumerate}
\end{proof}\\

\begin{proof}[Proof of Lemma \ref{lm:roc}]
From Lemma~\ref{lm:fund}, the characteristic equation associated with the nontrivial eigenvalues of the fundamental equilibrium is $\lambda^2-\gamma_1\lambda-\gamma_2=0$, with
\begin{linenomath}
$$
\gamma_1 = \Frac{1}{2R}\left((1+\bar{m}^{(0)})\hspace{0.2mm}v + 3(1-\bar{m}^{(0)})\right),\quad
\gamma_2 = -\Frac{1}{R}\hspace{0.2mm}(1-\bar{m}^{(0)}).
$$
\end{linenomath}
Note that the same characteristic equation is obtained for both predictors (\ref{eq:roc},b) and (\ref{eq:roc},c,d), and also for different choices of the confidence function \eqref{eq:roca}, as long as $\mathrm{\partial}\hspace{0.2mm}\alpha_{\text{ROC}}/\mathrm{\partial}\hspace{0.2mm}\text{ROC}|_{\text{ROC}=1}=0$.

The fundamental equilibrium is stable at $\beta=0$ (the Routh-Hurwitz-Jury test for 2nd-order polynomials requires $-\gamma_2|_{\beta=0}=1/R<1$ and $\gamma_1|_{\beta=0}=(v+3)/(2R)<1-\gamma_2|_{\beta=0}=(R+3+R-1)/(2R)$, which are readily satisfied).

As $\beta$ increases, transcritical and saddle-node (or pitchfork) bifurcations are not possible.
In fact, substituting $\lambda_1=1$ into the constraints:
\begin{linenomath}
$$
\lambda_1+\lambda_2=\gamma_1,\quad
\lambda_1\lambda_2=-\gamma_2,
$$	
\end{linenomath}
and eliminating $\lambda_2$, we get the contradiction
\begin{linenomath}
$$
R=\Frac{v}{2}\hspace{0.2mm}(1+\bar{m}^{(0)}) + \Frac{1}{2}\hspace{0.2mm}(1-\bar{m}^{(0)})
$$
\end{linenomath}
(with left-hand side larger than one and right-hand side smaller than 1).

Similarly, we exclude flip bifurcations: imposing $\lambda_1=-1$ in the above constraints and eliminating $\lambda_2$, we get the contradiction
\begin{linenomath}
$$
-R=\Frac{v}{2}\hspace{0.2mm}(1+\bar{m}^{(0)}) + \Frac{5}{2}\hspace{0.2mm}(1-\bar{m}^{(0)}).
$$
\end{linenomath}
(with left-hand side positive and right-hand side negative).

To look for a Neimark-Sacker bifurcation, we impose $\lambda_1\lambda_2=1$ and $|\lambda_1\lambda_2|<2$.
Under the first condition, the second turns into $v\hspace{0.2mm}(1+\bar{m}^{(0)})<R$ that is always satisfied (recall that $\bar{m}^{(0)}<0$, see Lemma~\ref{lm:fund}).
Solving the first condition for $\beta$ gives $\beta_{\mathrm{NS}}$.
\end{proof}

\bibliographystyle{plainnat}
\bibliography{C:/Users/UTENTE/Documents/PhD_Thesis_Davide_Radi/Thesis_ref}

\end{document}